\documentclass[]{article}
\usepackage{amsthm}
\usepackage{amsmath}
\usepackage{ amssymb }
\usepackage{amssymb}
\usepackage{dsfont}
\usepackage{ bbold }
\usepackage{pdflscape}
\usepackage{multirow}
\usepackage{pdfpages}
\usepackage{subcaption}
\usepackage{color}
\usepackage{verbatim}
\usepackage[round]{natbib} 
\usepackage[english]{babel}
\usepackage{float}
\usepackage{units}
\usepackage{mathrsfs}
\usepackage[margin=1.5in]{geometry}
\usepackage{makecell}
\usepackage{mathrsfs}

\newcommand{\R}{\mathbf{R}}
\newcommand{\x}{\mathbf{x}}
\newcommand{\tx}{\tilde{\mathbf{x}}}

\newcommand{\xx}{\mathbb{x}}
\newcommand{\zz}{\mathbb{z}}
\newcommand{\dmat}{\mathbf{d}}
\newcommand{\dtilde}{\tilde{\mathbf{d}}}
\newcommand{\dubd}{ \mathbf{d}{ \hspace{-.6 mm} \textnormal{l}}}
\newcommand{\dtildep}{\tilde{ \dmat }_{\hspace{-.6mm}{}{/}}{}_{\scriptscriptstyle \hspace{-.6mm}\matp}}

\newcommand{\tht}{{\scriptscriptstyle{\textnormal{HT}}}}
\newcommand{\tas}{{\scriptscriptstyle{\textnormal{AS}}}}
\newcommand{\tgn}{{\scriptscriptstyle{\textnormal{N}}}}

\newcommand{\thj}{{\scriptscriptstyle{\textnormal{HJ}}}}

\newcommand{\tl}{{\scriptscriptstyle{\textnormal{T}}}}
\newcommand{\tm}{{\scriptscriptstyle{\textnormal{M}}}}

\newcommand{\tcom}{{\scriptscriptstyle{\textnormal{CM}}}}
\newcommand{\tols}{{\scriptscriptstyle{\textnormal{OLS}}}}
\newcommand{\tiv}{{\scriptscriptstyle{\textnormal{IV}}}}
\newcommand{\twls}{{\scriptscriptstyle{\textnormal{WLS}}}}

\newcommand{\tb}{{\scriptscriptstyle{\textnormal{B}}}}
\newcommand{\tinv}{{\scriptscriptstyle{\textnormal{INVAR}}}}

\newcommand{\matp}{\mathbf{p}}
\newcommand{\m}{\mathbf{m}}

\newcommand{\W}{\mathbf{W}}

\newcommand{\w}{\mathbf{w}}

\newcommand{\I}{\mathbf{i}}
\newcommand{\E}{\text{\textnormal{E}}}
\newcommand{\V}{\text{\textnormal{V}}}
\newcommand{\diag}[1]{\textnormal{\footnotesize diag}\left(#1\right)}
\newcommand{\diagSmallBracket}[1]{\textnormal{\footnotesize diag}(#1)}
\newcommand{\nNegTwo}{n^{-2}}

\newcommand{\bfI}{\text{\textnormal{I}}}

\newcommand{\II}{{\scriptscriptstyle \text{\textnormal{I\hspace{-0.0mm}I}}}}
\newcommand{\sI}{{\scriptscriptstyle \text{\textnormal{I}}}}

\newcommand{\bpi}{\boldsymbol{\pi}}

\newcommand{\bpiInv}{\bpi^{-1}}

\usepackage{array}
\newcolumntype{L}[1]{>{\arraybackslash}p{#1}}
\newcolumntype{C}[1]{>{\centering \arraybackslash}p{#1}}
\newcommand{\onesmat}{\mathds{1}}

\usepackage{etex,etoolbox}
\usepackage{amsthm,amssymb}
\usepackage{thmtools}
\usepackage{environ}

\makeatletter
\providecommand{\@fourthoffour}[4]{#4}

\newcommand{\exampleType}[1]{\vspace{3mm}  \noindent \underline{Example} \hspace{1mm}(\textit{#1}) : \hspace{1mm}}
\providecommand{\@fourthoffour}[4]{#4}
\newcommand\fixstatement[2][\proofname\space of]{%
	\ifcsname thmt@original@#2\endcsname
	\AtEndEnvironment{#2}{%
		\xdef\pat@label{\expandafter\expandafter\expandafter
			\@fourthoffour\csname thmt@original@#2\endcsname\space\@currentlabel}%
		\xdef\pat@proofof{\@nameuse{pat@proofof@#2}}%
	}%
	\else
	\AtEndEnvironment{#2}{%
		\xdef\pat@label{\expandafter\expandafter\expandafter
			\@fourthoffour\csname #1\endcsname\space\@currentlabel}%
		\xdef\pat@proofof{\@nameuse{pat@proofof@#2}}%
	}%
	\fi
	\@namedef{pat@proofof@#2}{#1}%
}

\newcounter{proofcount}

\def\printproofs{%
	\count@=\z@
	\loop
	\the\toks\numexpr 1000 +\count@\relax
	\ifnum\count@<\value{proofcount}%
	\advance\count@\@ne
	\repeat}
\makeatother

\usepackage[affil-it]{authblk}
\declaretheorem[style=plain,name=Theorem, numberwithin=section]{theorem}

\newtheorem{lemma}[theorem]{Lemma}
\newtheorem{condition}{Condition}
\newtheorem{algorithm}[theorem]{Algorithm}
\newtheorem{remark}{Remark}
\newtheorem{definition}[theorem]{Definition}
\fixstatement[Proof of]{theorem}
\fixstatement[Proof of]{lemma}

\newcommand{\ones}[1]{ 1_{\scriptscriptstyle {#1}}}
\begin{document}

	\title{Unifying Design-based Inference: 
		\\ On Bounding and Estimating the Variance of 
		\\ any Linear Estimator in any Experimental Design
	\\	\small \vspace{5mm} WORKING PAPER 1 OF 4
	}

	\author{Joel A. Middleton{\footnote{
				Charles and Louise Travers Department of Political Science, \textit{University of California, Berkeley.} \\ \emph{{email:} joel.middleton@gmail.com}}} }
	
\date{April 4, 2021}

\maketitle

\pagebreak

\section{Introduction }

This paper provides a design-based framework for variance (bound) estimation in experimental analysis.  Results are applicable to virtually any combination of experimental design, linear estimator (e.g., difference-in-means, OLS, WLS) and variance bound, allowing for unified treatment and a basis for systematic study and compairison of designs using matrix spectral analysis. A proposed variance estimator reproduces Eicker-Huber-White (aka. ``robust", ``heteroskedastic consistent", ``sandwich", ``White", ``Huber-White", ``HC", etc.) standard errors and ``cluster-robust" standard errors as special cases. While past work has shown algebraic equivalences between design-based and the so-called ``robust" standard errors under some designs, this paper motivates them for a wide array of design-estimator-bound triplets. In so doing, it provides a clearer and more general motivation for ``robust" variance estimators.


\subsection{The Neyman Causal Model}

Consider a randomized experiment with $k$ treatment arms. The Neyman causal model (NCM) assumes that the units in the experimental study represent a finite population of size $n$. For a given outcome measure, call it $y$, each unit, $i$, responds with one of $k$ possible values in $\{y_{1i}, y_{2i}, ..., y_{ki} \}$, depending on their treatment assignment. The possible responses are referred to as the \textit{potential outcomes}. In the NCM these values are considered (nonrandom) constants, which stands in contrast to other, more common, formulations where potential outcomes are assumed to be sampled from some (possibly nonparametric) distribution. 

The only random element in the NCM is the treatment assignment indicators  $\{R_{1i}, R_{2i},...,R_{ki} \}$, and they determine which potential outcome will be observed by the researcher. Since a unit can only be assigned to one arm of the experiment, only one of the indicators will realize a value of one, and the rest will be zero, such that $R_{1i}+R_{2i}+...+R_{ki}=1$ for all $i$.

A standard representation of the \textit{observed} outcome for the $i^{th}$ unit under the NCM would be,
\begin{align*}
Y_i^{obs}= y_{1i}R_{1i}+y_{2i}R_{2i}+...+y_{ki}R_{ki},
\end{align*}
which is itself random, due to the assignment indicators. For each unit, the observed data can then be represented as $\{Y^{obs}_i, R_{1i}, R_{2i},..., R_{ki}, x_i\}_{\forall i}$, where $x_i$ is an additional vector of $l$ covariates. Like the potential outcomes, $x_i$ is nonrandom, but unlike the potential outcomes the same value is observed irrespective of the assignment.  

Ideally, we would like to know, for a given individual, $i$, the difference between responses under various arms, called a \textit{treatment effect}.  It is clear from the definition of $Y_i^{obs}$, however, that individual treatment effects are not observable since only one of the potential outcomes can be observed for an individual, a problem known as \textit{fundamental problem of causal inference} \citep{holland}. As a result, researchers often try to estimate \textit{averages} of across the units in study. 

\exampleType{Treatment/Control Experiment} In an experiment with a control group (arm 0) and a treatment group (arm 1) the individual-level treatment effect, $y_{1i}-y_{0i}$, but this is not identified, so a researcher might try to estimate the average treatment effect $n^{-1} \sum_i  \left(y_{1i}-y_{0i}\right)$.{ \hfill $\triangle$}

\exampleType{$2 \times 2$ Factorial Experiment} Consider a 2$\times$2 factorial design with treatments A and B. Units in arm 1 are controls (no treatments), units in arm 2 are given treatment A only, units in arm 3 are given B only, and units in arm 4 are given both A and B.  Similar to the treatment/control example, one could contrast the mean of an arm with a single treatment against the control mean, e.g., the average effect of A compared to no treatments, $n^{-1} \sum_i  \left(y_{2i}-y_{1i}\right)$.  Another quantity of interest might be an \textit{average marginal causal effect }(AMCE), e.g., the effect of A marginalizing over the levels of B, $n^{-1} \sum_i  \frac{1}{2} \left(y_{2i}-y_{1i}+y_{4i}-y_{3i}\right)$.  Another example might be an omnibus test based on the contrast $n^{-1} \sum_i \left(y_{2i}/3+y_{3i}/3+y_{4i}/3-y_{1i}\right)$.{ \hfill $\triangle$}

\vspace{2mm} Target quantities such as \textit{local average treatment effects} or \textit{conditional average treatment effects} might also be considered in this framework, but the primary focus of this paper is \textit{variance} estimation for linear estimators for virtually any design. 

Suffice to say that developing variance estimators before considering point estimation is appealing, if somewhat counter-intuitive, for two reasons. On the one hand, asymptotic analysis for point estimators can be made easier by having first established general variance expressions (for all linear estimators and virtually any design). On the other hand, 
a general framework for variance (bound) estimation can be developed even while a particular estimation target has yet to be defined, and even if an ``estimator" does not estimate anything of interest, it's variance can still be studied.

\subsection{Notation}

To simplify notation, let $y_1$, $y_2$,...,$y_k$ represent length $n$ vectors of potential outcomes associated with each of the arms, with the $i^{th}$ element of each corresponding to the $i^{th}$ unit.  Next, stack these vectors to create
\begin{align*}
y := \left(y_1' \hspace{2mm} y_2' \hspace{2mm} \hdots \hspace{2mm} y_k' \right)',
\end{align*}
which is a column vector and has length $kn$ containing all $k$ potential outcomes for all $n$ units.

Next, if we let $1_{\scriptscriptstyle n}$ be a $n$-length vector of ones, then a $kn \times k$ \textit{intercept matrix} can be defined as,
\begin{align*}
\onesmat := & \left[
\begin{matrix}
1_{\scriptscriptstyle n} &  & & 
\\  & 1_{\scriptscriptstyle n} &   & 
\\  &  & \ddots 
\\  &  & & 1_{\scriptscriptstyle n} 
\end{matrix} \right], 
\end{align*}
which, for example, allows us to express a \textit{k}-length vector of the means of each arm as $\frac{1}{n} \onesmat' y$, or, equivalently, $\left(\onesmat' \onesmat\right)^{-1} \onesmat' y$. Next, define $c$ as the \textit{contrast vector}, of length $k$, such that $c'\left(\onesmat' \onesmat\right)^{-1} \onesmat' y$ gives contrasts between potential outcome means for the various arms.

\exampleType{Treatment/Control Experiment, continued}
With two arms, control (arm 1) and treatment (arm 2), define $c=\left(-1 \hspace{2mm} 1\right)'$. Then the \textit{average treatment effect} is simply $c' \left(\onesmat' \onesmat\right)^{-1} \onesmat' y=  n^{-1} \sum_i  \left(y_{2i}-y_{1i}\right)$.{ \hfill $\triangle$}

\exampleType{$2 \times 2$ Factorial Experiment, continued} In a four-arm experiment, if $c=(-1\hspace{2mm} 1 \hspace{2mm} 0 \hspace{2mm} 0)'$ then $ c' \left(\onesmat' \onesmat\right)^{-1} \onesmat' y= n^{-1} \sum_i  \left(y_{2i}-y_{1i}\right) $ is the avearge difference between the first two arms. Alternatively, if the researcher chooses $c=(-\frac{1}{2}\hspace{2mm} \frac{1}{2} \hspace{2mm} -\frac{1}{2} \hspace{2mm} \frac{1}{2} )'$ then $c' \left(\onesmat' \onesmat\right)^{-1} \onesmat' y = n^{-1} \sum_i \frac{1}{2} \left(y_{2i}-y_{1i} + y_{4i}-y_{3i} \right)$. 
{ \hfill $\triangle$}

\vspace{2mm} Next define an $n \times n$ diagonal matrix that has all $n$ assignment indicators for treatment arm 1 on the diagonal, 
\begin{align*}
\R_1 :=&
\left[ \begin{matrix}
R_{11} \\ & R_{12} \\ & & \ddots \\& & &  R_{1i} \\ & & & & \ddots& \\ & & & & & R_{1n} 
\end{matrix}\right], \hspace{2mm}
\end{align*}
and define $\R_2$, $\R_3$, $\hdots$, $\R_k$ analogously. Arrange these matrices to create the diagonal $kn \times kn$ matrix
\begin{align*}
\R :=&
\left[ \begin{matrix}
\R_{1} \\ & \R_{2} \\ & & \ddots \\& & &  \R_{k}
\end{matrix}\right] \hspace{2mm}
\end{align*}
and note the a $kn\times kn$ diagonal matrix of assignment probabilities can be written as $\bpi:=\E[\R]$, with the first $n$ diagonal elements representing probabilities of assignment to arm 1, then the next $n$ diagonal elements are probabilities of assignment to arm 2 and so on. 

In this alternative notation the researcher can be said to observe the assignment, $\R$, the observed vector of outcomes, $\R y$, and also a matrix of $l$ pre-treatment covariates, $\x$, which has size $n \times l$. In a randomized experiment $\bpi$ is also observed (known) in many cases. When intractable analytically, however, it might be estimated to arbitrary precision by repeating the original randomization until a target level of precision is achieved.

For covariate adjusted estimators, it is convenient to define the $kn \times (k+l) $ matrix,
\begin{align*}
\xx := & \left[
\begin{matrix}
1_{\scriptscriptstyle n} & & & & \x
\\  & 1_{\scriptscriptstyle n} & & & \x
\\  &  & \ddots & & \vdots
\\  &  & & 1_{\scriptscriptstyle n} & \x
\end{matrix} \right]
\end{align*}
which augments the intercept vector, $\onesmat$, with covariates. 

\begin{remark}
	For some cases, such adjusting for covariates separately by arm, it might be useful to define $\xx$ with $\x$ matrices arranged along a block-diagonal. In that case, it is useful to stipulate that $\x$ have columns that sum to zero to avoid problems of coefficient interpretation \citep[cf.][]{lin, middleton18}. This will be discussed further in paper 3 of 4. 
\end{remark}


\section{Linear estimators}\label{section.estimators}

This paper focuses on the variance, bounding and variance bound estimation of the class of estimators that are linear in the observed outcome, $y$. This class includes everything from the difference-of-means, to the Horvitz-Thomposon estimator, to regression.  

Note, however, that beyond presenting a general approach to variance bound estimation for the class of linear estimators, point estimation itself will be the focus of the third and fourth papers in the series. Questions such as consistency will be and causal identification will be considered then.  For now, suffice it to be said that an estimator need not be consistent for any quantity of interest at all (causal or otherwise) in order to derive variance expressions for it.

\subsection{Definition}

\begin{definition}[Linear Estimators]\label{linear.est}
	Linear estimators are defined as having the form,
	\begin{align} 
	\widehat{\delta}_c := &  c' \W \R y,
	\end{align}
	where $\W$ a matrix with $kn$ columns and $k$ rows if it is an unadjusted estimator and $k+l$ rows if it is a covariate adjusted estimator. The length of the contrast vector, $c$, is equal to the number of rows in $\W$. The first $k$ entries of $c$ are the contrast values, followed by $l$ zeros in the case of covariate adjusted estimators.
\end{definition}


Also, for convenience, define $\w$ to be $\W$ evaluated at $\R=\bpi$, i.e.,  
\begin{align}
\w:=\{ \left. \W\right|_{\R=\bpi}\}.
\end{align}


\begin{definition}[Horvitz-Thompson estimator]
	The Horvitz-Thompson estimator written as in Definition \ref{linear.est} with,
	\begin{align*}
	\W &=  \W^\tht :=\left(\onesmat'\onesmat\right)^{-1} \onesmat'\bpiInv ,
	\end{align*}
	noting that $\W^\tht=\w^{\tht}$ since $\W^\tht$ is nonrandom.
\end{definition}

\begin{definition}[Contrast-of-means]
	Contrast-of-means (e.g., difference-of-means) can be written as in Definition 2.1 with,
	\begin{align*}
	\W &= \W^\tcom  := \left(\onesmat' \R \onesmat \right)^{-1}\onesmat' .
	\end{align*}
\end{definition}

\begin{definition}[Hajek estimator]
	The Hajek estimator can be written as Definition (\ref{linear.est}) with,
	\begin{align*}
	\W &= \W^\thj  := \left(\onesmat' \bpiInv \R \onesmat \right)^{-1}\onesmat' \bpiInv  .
	\end{align*}
\end{definition}

\begin{definition}[OLS estimator]
	The OLS estimator can be written as Definition (\ref{linear.est}) with,
	\begin{align*}
	\W &= \W^\tols  := \left(\xx' \R \xx \right)^{-1}\xx '  .
	\end{align*}
\end{definition}

\begin{definition}[WLS estimators]
	WLS estimators can be written as in Definition 2.1 with,
	\begin{align*}
	\W &= \W^\twls  := \left(\xx' \m \R \xx \right)^{-1}\xx' \m .
	\end{align*}
\end{definition}

\begin{remark}
	WLS is a class that includes OLS, Hajek and contrast-of-means (e.g., difference-of-means) as special cases.  It is equivalent to OLS when $\m=\I_{kn}$ ($\I_{kn}$ is the identity matrix). If $\m=\I_{kn}$ and, in addition, $\xx=\onesmat$, WLS is OLS without covariates, which is equivalent to the contrast-of-means (e.g., in the two-arm case, we call this the difference-of-means), underscoring Theorem 1 in \cite{freedman08a}.  If $\m=\bpiInv$ and $\xx=\onesmat$, then it is the Hajek estimator.  The covariate adjusted WLS with $\m=\bpiInv$ will be discussed further in paper 3 of 4, because it is algebraically equivalent to the generalized regression estimator introduced there.
\end{remark}

\subsection{First-order Taylor approximation}

In this section, a general approach to obtaining asymptotically valid variance expressions for linear estimators is given using a first-order approximation of a Tyalor series.  The method is often used when an exact, closed-form variance expression is not tractable, as may be the case with any number of linear estimators. Examination of the $W$ vectors defined above shows that, with the exception of Horvitz-Thompson, the estimators all had random denominators (i.e., inverted random matrices), making closed form variance expressions difficult.

The original estimator and its Taylor approximation are asymptotically equivalent (cite Pashley). As such, the original estimator ``borrows" the closed-form variance expression given for the Taylor approximation, again justified given the asymptotic equivalence. 

\begin{lemma}[First-order Taylor approximation for linear estimators]
First, assume a linear estimator as defined in Definition 2.1. Then, let $\Big \{  \left. . \hspace{1mm} \right| _{\R = \bpi} \Big \}$ represent a function that evaluates the argument to the left of the vertical line at ${\R = \bpi}$. Similarly, let $\Big \{  \left. . \hspace{1mm} \right| _{\R = \bpi} \left(\R-\bpi\right)\Big \}$ evaluate its argument at ${\R = \bpi}$ and then multiply by $\left(\R-\bpi\right)$. Then from Taylor's theorem and the product rule, we have the first-order Taylor approximation, $\widehat{\delta} \approx \widehat{\delta}^\tl$, with 
\begin{align}\label{TaylorLinearization}
\widehat{\delta}^\tl_c
:= & \Big\{ \left. c' \W \R y \hspace{1mm} \right|_{\vspace{10mm }\scriptstyle \R =\bpi}\Big \}  
+\Big\{  \left. c' \W \right|_{\scriptstyle \R =\bpi} \Big \}\bigg \{  \left. \frac{\text{d}}{\text{d} \R} \R \right|_{\scriptstyle \R =\bpi} \left(\R-\bpi\right) \bigg \}  y \nonumber
\\ &\hspace{17mm}+ \bigg \{ \left. \frac{\text{d}}{\text{d} \R} c' \W \right|_{\vspace{10mm }\scriptstyle \R =\bpi} \left(\R-\bpi\right) \bigg\}  \Big\{  \left. \R \right|_{\scriptstyle \R = \bpi } \Big \}y \nonumber
\\ =& \hspace{2mm} a_c + c' \w  \R  y + \bigg \{ \left. \frac{\text{d}}{\text{d} \R} c' \W \right|_{\vspace{10mm }\scriptstyle \R =\bpi} \R \bigg\}  \bpi y  
\end{align}
where
\begin{align*}
a_c =&   - \bigg \{ \left. \frac{\text{d}}{\text{d} \R} c' \W \right|_{\vspace{10mm }\scriptstyle \R =\bpi} \bpi \bigg\} \bpi  y
\end{align*}
is a constant.
\end{lemma}
\begin{remark}
An expression for $a_c$ is given but it is not important for the purposes of variance approximations because the term is a constant. Recall that the purpose of deriving a first-order Taylor approximation, $\widehat{\delta}^\tl_c$, is to identify a closed-form variance expression that might then be ``borrowed" by the original linear estimator given in Definition \ref{linear.est}. 
\end{remark}

\begin{theorem}\label{theorem.Taylor.is.HT}
	For a constant, $a_c$, and vector of constants, $z_c$, first-order Taylor approximations for linear estimators may be written as,
		\begin{align*}
		\widehat{\delta}^\tl_c 
		 = & \hspace{1mm} a_c 
		+ n \ones{k}' \w^{\tht} \R z_c,
		\end{align*}
where $z_c$ has the form $z_c = \bpi \diag{ \mathbf{l} y} \mathbf{t}' c$ and where $(k \times kn)$ matrix $\mathbf{t}$ and $(kn \times kn)$ matrix $\mathbf{l}$ depend on the estimator. Hence, a first-order approximation of a Tyalor series using Taylor's theorem variance approximations will be expressed as the variance of a Horvitz-Thompson estimator of the ATE of $z_c$ with contrast vector $n\ones{k}$.
\end{theorem}
\begin{proof}
    With Equation (\ref{TaylorLinearization}), it is easy to see that the Taylor linearized approximation has the form
    \begin{align*}
       \widehat{\delta}^\tl_c= & a_c +c' \mathbf{t} \R \mathbf{l} y
    \end{align*}
    where matrices $\mathbf{t}$ and $\mathbf{l}$ are $(k \times kn)$ and $(kn \times kn)$, respectively, and will depend on the estimator. Noting that $c' \mathbf{t}$ is a $(1 \times kn)$ vector, write 
    \begin{align*}
       \widehat{\delta}^\tl_c = & a_c +\ones{kn}' \diag{ c' \mathbf{t} } \R \mathbf{l} y
       \\ = & a_c +\ones{kn}' \R \diag{ c' \mathbf{t} } \mathbf{l} y
       \\ = & a_c +\ones{kn}' \bpiInv \R \bpi \diag{ c' \mathbf{t} } \mathbf{l} y
      \\ = & a_c + n \ones{k}' (\onesmat' \onesmat)^{-1} \onesmat' \bpiInv \R \bpi \diag{ c' \mathbf{t} } \mathbf{l} y
      \\ = & a_c + n \ones{k}' \w^\tht \R z_c
    \end{align*}
where $z_c := \bpi \diag{ \mathbf{l} y} \mathbf{t}' c$.
\end{proof}

\begin{remark}
	The result shows that first order Taylor approximations are Horvitz-Thompson estimators. This highlights the importance of studying Horvitz-Thompson variance in order to develop asymptotic variance expressions for linear estimators in general.
\end{remark}
\begin{remark}
	The constant vector $z_c$ is not directly observed. The next section will show that the plug-in principle provides a basis for asymptotically valid variance expressions.
\end{remark}

\setlength{\extrarowheight}{12pt}
\begin{table}[H]
	\begin{center} 
		\caption{Examples of linear estimators. $\W$ is as defined in Definition 2.1, $z_c$ is as defined in Theorem \ref{theorem.Taylor.is.HT}.}
		\begin{tabular}{l | c | c }\label{table.examples.linear.ests}
			Estimator &  $\W$ & $z_c$   
			\\ 	\hline Horvitz-Thompson 
			&	$  \left(\onesmat' \onesmat\right)^{-1} \onesmat' \bpiInv$
			&  $ \diagSmallBracket{y} \onesmat \left(\onesmat' \onesmat\right)^{-1} \hspace{-1mm} c$
			\\  \makecell[l]{Contrast-of-means}  &  $ \left(\onesmat' \R \onesmat\right)^{-1} \onesmat' $ 
			& $  \bpi \diagSmallBracket{y - \onesmat\left( \onesmat' \bpi \onesmat\right)^{-1} \onesmat' \bpi y} \onesmat \left( \onesmat' \bpi \onesmat\right)^{-1} \hspace{-1mm} c  $ 
			\\ Hajek &  $ \left(\onesmat' \bpiInv \R \onesmat\right)^{-1}\onesmat' \bpiInv $ 
			& $ \diagSmallBracket{y - \onesmat\left( \onesmat' \onesmat\right)^{-1} \onesmat'y} \onesmat  \left(\onesmat' \onesmat\right)^{-1}\hspace{-1mm} c $  
			\\ 
			OLS & $\left(\xx' \R \xx \right)^{-1}\xx'$ 
			& $ \bpi \diag{y- \xx b^{\tols}} \xx \left( \xx' \bpi \xx\right)^{-1} \hspace{-1mm} c $
			\\ WLS & $  \left(\xx' \m \R \xx \right)^{-1} \xx' \m $  
			& $ \bpi \diag{y- \xx b^{\twls}} \m  \xx \left( \xx' \m  \bpi \xx\right)^{-1} \hspace{-1mm} c $ 
			\\ \makecell[l]{Generalized reg. \\ \hspace{1mm} ($b=b^{\twls}$)} & $\w^\tht \left(\I_{kn}- \left( \R -\bpi \right)\xx \W^\twls\right) $ 
			& $\diag{y- \xx b^{\twls} } \onesmat \left(\onesmat' \onesmat\right)^{-1} \hspace{-1mm} c$ 
			\\ \makecell[l]{ IV } 
			&  \makecell{$  \big (\widehat{\tilde{\xx}}' \R \widehat{\tilde{\xx}} \big )^{-1} \widehat{\tilde{\xx}}' $ \hspace{2mm}with: 
			 \\ $\widehat{\tilde{\xx}}:= \zz \left( \zz' \R \zz\right)^{-1} \zz' \R \xx $
			}
			& \makecell{$ \bpi \diag{y- \xx b^{\tiv}} \tilde{\xx} \left( \tilde{\xx}' \bpi \tilde{\xx}\right)^{-1} \hspace{-1mm} c $ \hspace{2mm}with: 
			\\ $\tilde{\xx}:= \zz \left( \zz' \bpi \zz\right)^{-1} \zz' \bpi \xx$,  $b^{\tiv}:= \left( \tilde{\xx}' \bpi \tilde{\xx}\right)^{-1} \tilde{\xx}' \bpi y$ } 
	\end{tabular} \end{center}
\end{table}
\setlength{\extrarowheight}{0pt}

\exampleType{Weighted least squares} Weighted least squares is a class that includes OLS ($\m=\I_{kn}$), contrast-of-means (e.g., difference of means, with $\m=\I_{kn}$ and $\xx=\onesmat$) and the Hajek estimator ($\m=\bpiInv$ and $\xx=\onesmat$). To derive its Taylor approximation, first let $\w^\twls= {\left. {\W^{\twls}} \right| }_{\R=\bpi}= \left( \xx' \m \bpi \xx\right)^{-1} \xx' \m $, and note that by the rules of matrix differentiation the third term in Equation (\ref{TaylorLinearization}) is
\begin{align*}
    \bigg \{  \left. \frac{\text{d}}{\text{d} \R} c' \left( \xx' 
    \m \R \xx \right)^{-1}\xx' 
    \m  \right|_{\vspace{10mm }  \scriptstyle \R =\bpi}  \R  \bigg\}  & \bpi y  
    \\ =  -c' & \left( \xx' \m \bpi \xx \right)^{-1}  \bigg \{ \left. \frac{\text{d}}{\text{d} \R} \left(\xx' 
    \m \R \xx \right) \right|_{\vspace{10mm }\scriptstyle \R =\bpi} \R \bigg\} \left( \xx' \m \bpi \xx \right)^{-1}\xx' \m \bpi y
    \\ =  -c' & \left( \xx' \m \bpi \xx \right)^{-1}  \xx' \m \bigg \{ \left. \frac{\text{d}}{\text{d} \R}  
    \R \right|_{\vspace{10mm }\scriptstyle \R =\bpi} \R \bigg\} \xx   b^\twls 
    \\ =  -c' & \w^\twls  \R \xx   b^\twls. 
\end{align*}
Therefore, Equation (\ref{TaylorLinearization}) made specific to WLS is  
\begin{align*}
    \widehat{\delta}^{\tl (\twls)}= & a_c^\twls + c' \w^{\twls}\R y  -\w^\twls  \R \xx   b^\twls
    \\ = & a_c^\twls + c' \w^{\twls}\R \left( y-\xx   b^\twls\right)
    \\ = & a_c^\twls + \ones{kn}' \diag{c' \w^{\twls}} \R \diag{y-\xx   b^\twls} \ones{kn}
    \\ = & a_c^\twls + \ones{kn}' \R  \diag{y-\xx   b^\twls} {\w^{\twls}}' c
    \\ = & a_c^\twls + n\ones{k}' \w^\tht \R  z_c^\twls
\end{align*}
where $z_c^{\twls} = \bpi \diag{y- \xx b^{\twls}} {\w^\twls}' c $ is recognizable in the form given in Theorem \ref{theorem.Taylor.is.HT} with $\textbf{t}^\twls=\w^{\twls}$ and $\textbf{l}^{\twls} = \I_{kn} -\xx \left( \xx' \m \bpi \xx \right)^{-1} \xx' \m \bpi $ is a ``residual maker" matrix. 
{ \hfill $\triangle$}

\section{Variance}\label{section.var}

Now that the importance of Horvitz-Thompson estimators for asymptotic variance expressions for the entire class of linear estimators (which includes, for example, OLS, WLS, Hajek, and difference-of-means) has been established, this section will give the variance of HT estimators and first-order approximates of linear estimators.  

Throughout, we will make use of the $kn \times kn$ ``first order design matrix", which will allow for easy comparison of designs using spectral analysis.

\begin{definition}
The ``first-order design matrix" is a variance-covariance matrix of inverse-probability weighted treatment assignments, written,
\begin{align}\label{dmat}
	\dmat:=&\V \left(\ones{kn}'\bpiInv \R \right)
	\\ = & \left(\E\left[\R \ones{kn} \ones{kn}' \R \right] - \bpi \ones{kn} \ones{kn}' \bpi \right) / \left( \bpi \ones{kn} \ones{kn}' \bpi \right), \nonumber
\end{align}
where ``/" represents elementwise division.
\end{definition}


\begin{theorem}[Horvitz-Thompson Variance]\label{theorem.HTvar} An exact expression for the variance of Horvitz-Thompson estimators is given by 
    \begin{align*}
    \V\left(\widehat{\delta}^{\tht}_c \right) =  { z^{\tht}_c}'\dmat { z^{\tht}_c},
    \end{align*}
    where ${z_c^{\tht}}' = c' \left(\onesmat' \onesmat \right)^{-1} \onesmat' \diag{y}$.
\end{theorem}
\begin{proof}
    Using the identity $y=\diag{y} \ones{kn}$, the Horvitz-Thompson estimator can be written
    \begin{align*}
        \widehat{\delta}^{\tht}
        = &  c' \left(\onesmat' \onesmat \right)^{-1} \onesmat' \bpiInv \R \diag{y} \ones{kn}
        \\         = &  c' \left(\onesmat' \onesmat \right)^{-1} \onesmat'\diag{y} \bpiInv \R  \ones{kn}.
                \\         = &  { z^{\tht}_c}' \bpiInv \R  \ones{kn}.
    \end{align*}
    where ${z_c^{\tht}}' = c' \left(\onesmat' \onesmat \right)^{-1} \onesmat' \diag{y}$. 
    So the variance can be written,
        \begin{align*}
    \V\left(\widehat{\delta}^{\tl}_c \right) & =  { z^{\tht}_c}'\V \left(\ones{kn}'\bpiInv \R \right) { z^{\tht}_c}
    \\ & = { z^{\tht}_c}'\dmat { z^{\tht}_c}
    \end{align*}
\end{proof}

\begin{theorem}[Variance of first-order Taylor approximations] The variance of first-order Taylor approximations of linear estimators can be written as,
\begin{align} \label{var.general}
\V\left(\widehat{\delta}^{\tl}_c \right) =  {z_c^\tl}' \dmat z_c^\tl,
\end{align}
with examples of $z^\tl_c$ given in Table \ref{table.examples.linear.ests}. 
\end{theorem}
\begin{proof}
    By Theorem \ref{theorem.Taylor.is.HT}, linear approximations are Horvitz-Thompson estimators of a vector $z_c$ and contrast vector $n\ones{k}$. Now, $ n \ones{k}' \left( \onesmat' \onesmat \right)^{-1} \onesmat' \diag{z^\tl_c}= \ones{kn}' \diag{z^\tl_c} ={z^\tl_c} $. Therefore, using Theorem \ref{theorem.HTvar},
    \begin{align*}
    \V\left(\widehat{\delta}^{\tl}_c \right) & = n \ones{k}' \left(\onesmat' \onesmat \right)^{-1} \onesmat' \diag{z^\tl_c} \dmat  \diag{z^\tl_c} \onesmat \left(\onesmat' \onesmat \right)^{-1} \ones{k} n 
    \\ & = {z^\tl_c}' \dmat {z^\tl_c}
    \end{align*}
 
\end{proof}

\begin{remark}
	Equation (\ref{var.general}) is an \textit{exact} variance expression, however, they are not identified. The next subsection introduces the necessary concept of variance bounding.
\end{remark}

\begin{remark}
	The design matrix, $\dmat$, will provide useful device in the study the best designs for outcomes with different characteristics as the next example will show.
\end{remark}


\exampleType{Comparing complete randomization and paired randomization}  Consider a treatment/control (two-arm) experiment that is pair-randomized. A pair-randomized design is a special case of a block-randomized (i.e., stratified) design where blocks have size 2.  In each pair/block, one unit is assigned to treatment and the other in control with equal (.5) probability. Across blocks, assignments are independent. 

When $n=4$ (and assuming w.l.o.g. that the data are sorted by pair), the design matrix is
\begin{align*}
\dmat^{pr} = \left[\begin{smallmatrix}
\hspace{2.2mm}1 & -1 & & & -1 & \hspace{2.2mm}1 &  
\\ -1 & \hspace{2.2mm}1 & & & \hspace{2.2mm}1 & -1
\\ & & \hspace{2.2mm}1 & -1 & & & -1 & \hspace{2.2mm}1
\\ & & -1 & \hspace{2.2mm}1 & & & \hspace{2.2mm}1 & -1
\\ -1 & \hspace{2.2mm}1 & & & \hspace{2.2mm}1 & -1
\\ \hspace{2.2mm}1 & -1 & & & -1 & \hspace{2.2mm}1 & 
\\ & & -1 & \hspace{2.2mm}1 & & & \hspace{2.2mm}1 & -1
\\ & & \hspace{2.2mm}1 & -1 & & & -1 & \hspace{2.2mm}1 & 
\end{smallmatrix}\right],
\end{align*}
and note that empty cells are 0. The design matrix for complete randomization (where 2 of 4 are randomly assigned to treatment) is
\begin{align*}
\dmat^{cr} = \left[\begin{smallmatrix}
\hspace{1mm}1 & \nicefrac{\text{-}1}{3} & \nicefrac{\text{-}1}{3} & \nicefrac{\text{-}1}{3} & 
  \text{-}{1} & \hspace{1mm}\nicefrac{1}{3} &  \hspace{1mm}\nicefrac{1}{3}& \hspace{1mm}\nicefrac{1}{3}
\\ \nicefrac{\text{-}1}{3}  & \hspace{1mm}1 & \nicefrac{\text{-}1}{3}  & \nicefrac{\text{-}1}{3}  & 
     \hspace{1mm}\nicefrac{1}{3} & \text{-}1 & \hspace{1mm}\nicefrac{1}{3} & \hspace{1mm}\nicefrac{1}{3} 
\\ \nicefrac{\text{-}1}{3}  & \nicefrac{\text{-}1}{3}  & \hspace{1mm}1 & \nicefrac{\text{-}1}{3}  & 
    \hspace{1mm}\nicefrac{1}{3}& \hspace{1mm}\nicefrac{1}{3}& \text{-}1 & \hspace{1mm}\nicefrac{1}{3}
\\ \nicefrac{\text{-}1}{3} & \nicefrac{\text{-}1}{3} & \nicefrac{\text{-}1}{3} & \hspace{1mm}1 &\hspace{1mm}\nicefrac{1}{3} & \hspace{1mm}\nicefrac{1}{3}& \hspace{1mm}\nicefrac{1}{3} & \text{-}1
\\ \text{-}1 & \hspace{1mm}\nicefrac{1}{3} & \hspace{1mm}\nicefrac{1}{3} & \hspace{1mm}\nicefrac{1}{3} & 
    \hspace{1mm}1 & \nicefrac{\text{-}1}{3} & \nicefrac{\text{-}1}{3} & \nicefrac{\text{-}1}{3} 
\\ \hspace{1mm}\nicefrac{1}{3} & \text{-}1 & \hspace{1mm}\nicefrac{1}{3} & \hspace{1mm}\nicefrac{1}{3}   & 
  \nicefrac{\text{-}1}{3}  & \hspace{1mm}1 & \nicefrac{\text{-}1}{3}  & \nicefrac{\text{-}1}{3}  & 
\\  \hspace{1mm}\nicefrac{1}{3} &  \hspace{1mm}\nicefrac{1}{3} & \text{-}1  & \hspace{1mm}\nicefrac{1}{3} & 
  \nicefrac{\text{-}1}{3}  & \nicefrac{\text{-}1}{3}  & \hspace{1mm}1 & \nicefrac{\text{-}1}{3}  & 
\\ \hspace{1mm}\nicefrac{1}{3} & \hspace{1mm}\nicefrac{1}{3} & \hspace{1mm}\nicefrac{1}{3}  & \text{-}1 & 
    \nicefrac{\text{-}1}{3} & \nicefrac{\text{-}1}{3} & \nicefrac{\text{-}1}{3} & \hspace{1mm}1 
\end{smallmatrix}\right].
\end{align*}
Eigendecomposition of $\dmat^{cr}-\dmat^{pr}$ gives eigenvalues $2.67, 0, 0, 0, 0, 0, -1.33,$ and $-1.33$ corresponding eigenvectors in Table \ref{table.eigvecs}.  The eigenvectors associated with nonzero eigenvalues provide insight into the subspace in $ {\mathds {R}}^{\scriptscriptstyle 2n}$ where one design may be preferable to another, for example, when the estimator is difference-in-means (which is equivalent to both Horvitz-Thopson and Hajek for these designs).  

\begin{table}[ht]
\begin{center}\caption{Eigenvectors of $\dmat^{cr}-\dmat^{pr}$}
	\begin{tabular}{rrrrrrrr}\label{table.eigvecs}
		e1 \hspace{1mm} & e2 \hspace{1mm} & e3 \hspace{3mm}& e4\hspace{2mm} & e5\hspace{2mm} & e6\hspace{2mm} & e7\hspace{2mm} & e8\hspace{2mm} \\ 
		\hline
		-0.354 & 0.791 & 0.000 & 0.000 & 0.000 & 0.000 & 0.500 & 0.000 \\ 
		 -0.354 & 0.158 & -0.573 & -0.178 & -0.250 & -0.421 & -0.500 & 0.000 \\ 
		0.354 & 0.158 & -0.180 & -0.450 & 0.585 & -0.149 & 0.000 & 0.500 \\ 
		0.354 & 0.158 & 0.319 & -0.600 & -0.260 & -0.264 & 0.000 & -0.500 \\ 
		0.354 & 0.474 & 0.282 & 0.524 & 0.100 & -0.190 & -0.500 & 0.000 \\ 
		0.354 & -0.158 & -0.291 & 0.346 & -0.150 & -0.611 & 0.500 & 0.000 \\ 
		-0.354 & -0.158 & 0.102 & 0.073 & 0.685 & -0.339 & -0.000 & -0.500 \\ 
		-0.354 & -0.158 & 0.602 & -0.077 & -0.161 & -0.454 & -0.000 & 0.500 \\ 
		\hline
	\end{tabular}
\end{center}
\end{table}

In this example, assuming the contrast matrix is $c=\left(-1,1\right)'$, and examining the first eigenvector with eigenvalue 2.67, one can conclude that if the outcomes for the four units given in Table \ref{table.bestcase.cr}, then the difference-of-means would be much less precise under the completely randomized design. So, the eigenvector in a sense represents a ``best-case" (normed) potential outcome vector for paired randomization.  Inspection of the outcomes themselves confirms the intuition that pair randomization is better than complete randomization when units are homogenous within pairs.  
\begin{table}[H]
	\begin{center}\caption{Pair randomization better than complete randomization}
\begin{tabular}{c | c | c  c}\label{table.bestcase.cr}
	unit id & pair id & $y_0$ & $y_1$  \\ \hline
		 1 & 1 &\hspace{1mm}.3536 & \hspace{1mm}.3536  \\
		 2 & 1 &\hspace{1mm}.3536 & \hspace{1mm}.3536  \\
		 3 & 2 &-.3536 & -.3536  \\
		 4 & 2 & -.3536 & -.3536  
\end{tabular}
	\end{center} 
\end{table}
Next, considering the two eigenvectors associated with the eigenvalue -1.33, we see the implied potential outcomes in Table \ref{table.bestcase.pr} give potential outcomes for which complete randomization is preferable.  Note that either of the two sets is a ``worst-case" scenario for paired randomization, as is any set of potential outcomes that can be generated by linear combinations of the two eigenvectors. Inspection of these outcomes is consistent with the observation that complete randomization can be better than paired randomization when paired units are maximally heterogeneous.
\begin{table}[H]
	\begin{center}\caption{Complete randomization better than pair randomization}
		\begin{tabular}{c | c | c  c | c c}\label{table.bestcase.pr}
			unit id & pair id & $y_0$ & $y_1$ & $y_0$ & $y_1$  \\ \hline
			1 & 1 & -.5 & -.5  &\hspace{2mm}0 & \hspace{2mm}0   \\
			2 & 1 &\hspace{1mm}.5 & \hspace{1mm}.5 &\hspace{2mm}0 & \hspace{2mm}0   \\
			3 & 2 &\hspace{2mm}0 & \hspace{2mm}0 & -.5 & -.5   \\
			4 & 2 & \hspace{2mm}0 & \hspace{2mm}0 &\hspace{1mm}.5 & \hspace{1mm}.5 
		\end{tabular}
	\end{center} 
\end{table}
{ \hfill $\triangle$}

\begin{remark}
	The example illustrates a relatively effortless method of identifying key insights about arbitrary designs through spectral analyses of first-order design matrices. In the example, the observation that pair randomization can hurt precision when units are not homogeneous within pairs is not new. However, this approach to comparing designs is perfectly general and can be applied to virtually any designs.  
\end{remark}

\section{Variance bounds}\label{section.var.bounds}

In spite of an exact expression for first-order Taylor approximations in Equation (\ref{var_lin_est}), the quantity is never identified because not all terms in the quadratic can be observed. Even if the elements of $\R z_c$ were observed directly (which is the case for $\R  z_c^\tht$ but none of the other examples in Table 1), some pairs of potential outcomes can never be jointly observed. For example, for a given unit, only one of two (or more) potential outcomes can be observed, a problem is referred to as the ``fundamental problem of causal inference" \citep{holland}.  Other design features, such as clustering or pair randomization, can also render various combinations of potential outcomes unobservable.

Starting with Neyman (1923) one proposed solution to unidentified variance has been to estimate a {\it variance bound}, i.e., a quantity that is provably greater than the variance, but which is identified. It should be understood that while the term \textit{variance estimation} is often used as a shorthand in the literature, it is not, in general, an accurate phrase.  \textit{Variance bound estimation} is a more precise so it will be used here.

\begin{definition}[Variance bound matrix]\label{def.varbound}
	Let $\tilde{\dmat}$ be an arbitrary $kn \times kn$ matrix and let be $z$ an arbitrary vector with length $kn$. Then $\tilde{ \dmat }$ is a \textnormal{variance bound matrix} (or \textnormal{bounding matrix}) for $\dmat$ if, for all $z\in\mathds{R}^{kn}$, $z'\dmat z \leq  z'\tilde{\dmat}z$. 
\end{definition}

\begin{lemma}\label{psd}
	$\tilde{ \dmat }$ is a bounding matrix $\dmat $ if and only if matrix $\tilde{ \dmat }-\dmat$ is positive semi-definite. 
\end{lemma}

\begin{proof}
	By the definition of a bound, $ z'\tilde{\dmat}z - z'\dmat z \geq 0$ for all $z\in \mathds{R}^{kn}$. This implies that $z'(\tilde{\dmat}-\dmat ) z \geq 0$, i.e., that $ \tilde{\dmat}-\dmat  $ is positive semi-definite.  
\end{proof}

\begin{definition}[Identified variance bound]\label{def.identified.bound}
	Let $\tilde{\dmat}$ be bounding matrix for $\dmat$. It gives an \textnormal{identified variance bound} if 
	\begin{align*}
	\bfI (\dmat=-1) 
	\circ \bfI(\tilde{\dmat}= 0)=\bfI (\dmat=-1)
	\end{align*}
	where $\circ$ is element-wise multiplication, $\bfI\left(\dmat=-1\right)$ is an indicator function returning an $kn \times kn$ matrix of ones and zeros indicating whether each element of $\dmat$ is equal to $-1$ (an indication that the associated term in the variance quadratic is impossible to observe), and $\bfI (\tilde{ \dmat } =0 )$ is, similarly, an indicator function returning an $kn \times kn$ matrix of ones and zeros indicating the location of zeros in $\tilde{ \dmat }$.	
\end{definition}


\subsection{Generalizing Neyman's variance bound}\label{section.GNbound}

This section proposes a generalization of Neyman's (1923) variance bound. Let matrix $\dmat$ be partitioned into $k^2$ partitions of size $n \times n$. Then for $r,s \in \{1, 2, ..., k\}$, let the $\dmat_{rs}$ be the $(r,s)^{th}$ partition, having dimension $n \times n$.  Also, let $c_r$ be the $r^{th}$ element of the length-$k$ contrast vector, $c$.
Then the following bounding method produces an identified bound for experiments when partitions $\textnormal{I}(\tilde{\dmat}^\tgn_{rr}==-1)=0_{\scriptscriptstyle n\times n}$, i.e., there are no $-1$ values in the diagonal blocks, and $\tilde{\dmat}^\tgn_{rs}=\tilde{\dmat}^\tgn_{tu}$ for $r\neq s, t\neq u \in {1,2,...,k}$ and $\sum_i c_i = 0$. Designs that meet this condition include complete randomization, cluster-randomization and block-randomization.
\begin{definition}[Generalized Neyman variance bound]
	The ``Generalized Neyman bound" is the is the bound corresponding to the block-diagonal bounding matrix, $\dtilde^{\tgn}$, with block $(r,r)$ given by,
	\begin{align*}
	\dtilde^\tgn_{rr} :=& \sum_{s=1}^{k} \frac{c_r}{c_s}\dmat_{rs} 
	\end{align*}
where $c_r$ and $c_s$ are, respectively, elements $r$ and $s$ from from the contrast vector, $c$.  
\end{definition}\label{definition.neyman.bound}

\begin{theorem}
The generalized Neyman bound, with $\tilde{\dmat}^\tgn$ given in Definition \ref{definition.neyman.bound} is an identified variance bound when partitions $\textnormal{I}(\tilde{\dmat}^\tgn_{rr}==-1)=0_{\scriptscriptstyle n\times n}$, i.e., there are no $-1$ values in the diagonal blocks, $\tilde{\dmat}^\tgn_{rs}=\tilde{\dmat}^\tgn_{tu}$ for $r\neq s, t\neq u \in {1,2,...,k}$, and $\sum_i c_i = 0$.
\end{theorem}

\begin{proof}
First, with $z_c^\tht =\diagSmallBracket{y} \onesmat (\onesmat' \onesmat){}^{-1}$ and letting $y_r$ be the length-$n$ vector of potential outcomes for the $r^{th}$ treatment arm and $\tilde{\dmat}^\tgn_{rs}$ be the $r,s$ partition of $\dmat$, we have
\begin{align*}
n^2 {z^{\tht}}' \dmat z^\tht = & y' \diagSmallBracket{c'\onesmat} \dmat \diagSmallBracket{c'\onesmat} y
\\ = & \sum_{r=1}^k c_r^2 y_r' \dmat_{rr} y_r +\sum_{r=1}^{k-1} \sum_{s=1}^{k}  c_r c_s \left(y_r' \dmat_{rs} y_s+y_s' \dmat_{sr} y_r\right)
\end{align*}
Next, define the $r,s$ treatment effect as $\tau_{rs}:=y_r-y_s$ and note that $\dmat_{12}=\dmat_{rs}$ for $r\neq s$. Then, by the definition of $\tilde{\dmat}^\tgn$, 
\begin{align*}
n^2 {z^{\tht}}' \dtilde^\tgn z^\tht = & y' \diagSmallBracket{c'\onesmat} \dtilde^\tgn \diagSmallBracket{c'\onesmat} y
\\ = & \sum_{r=1}^k c_r^2 y_r' \dmat_{rr} y_r +\sum_{r=1}^{k-1} \sum_{s=r+1}^{k}  c_r c_s \left(y_r' \dmat_{rs} y_r+y_s' \dmat_{sr} y_s\right)
\\ = & \sum_{r=1}^k c_r^2 y_r' \dmat_{rr} y_r +\sum_{r=1}^{k-1} \sum_{s=r+1}^{k}  c_r c_s \left(y_r' \dmat_{12} y_r+y_s' \dmat_{12} y_s\right)
\\ = & \sum_{r=1}^k c_r^2 y_r' \dmat_{rr} y_r +\sum_{r=1}^{k-1} \sum_{s=r+1}^{k}  c_r c_s \left(y_r' \dmat_{12} (y_s+\tau_{rs})+y_s' \dmat_{12} (y_r-\tau_{rs}) \right)
\\ = & n^2 {z^{\tht}}' \dmat z^\tht + \sum_{r=1}^{k-1} \sum_{s=r+1}^{k}  c_r c_s \left(y_r' \dmat_{12}\tau_{rs}- y_s'\dmat_{12} \tau_{rs} \right)
\\ = & n^2 {z^{\tht}}' \dmat z^\tht + \sum_{r=1}^{k-1} \sum_{s=r+1}^{k}  c_r c_s \left(y_r' \dmat_{12}\tau_{rs}- (y_r-\tau_{rs})'\dmat_{12} \tau_{rs} \right)
\\ = & n^2 {z^{\tht}}' \dmat z^\tht + \sum_{r=1}^{k-1} \sum_{s=r+1}^{k}  c_r c_s \tau_{rs}'\dmat_{12} \tau_{rs}.
\end{align*}
Next, to show that the second term is non-negative, note that $\tau_{rs} =\tau_{rk}-\tau_{sk}$, and write
\begin{align*}
    \sum_{r=1}^{k-1} \sum_{s=r+1}^{k}  c_r c_s \tau_{rs}'\dmat_{12} \tau_{rs} = & \frac{1}{2}\sum_{r=1}^{k} \sum_{s=1}^{k}  c_r c_s \tau_{rs}'\dmat_{12} \tau_{rs}  
    \\ = & \frac{1}{2}\sum_{r=1}^{k} \sum_{s=1}^{k}  c_r c_s \left(\tau_{rk}-\tau_{sk}\right)'\dmat_{12} \left(\tau_{rk}-\tau_{sk}\right) \\ = & \frac{1}{2}\sum_{r=1}^{k} \sum_{s=1}^{k}  c_r c_s \left(\tau_{rk}'\dmat_{12} \tau_{rk}+\tau_{sk}'\dmat_{12} \tau_{sk} - 2 \tau_{sk}'\dmat_{12} \tau_{rk}\right)   
    \\ = & \sum_{r=1}^{k} \sum_{s=1}^{k}  c_r c_s \tau_{rk}'\dmat_{12}  \tau_{rk} -\sum_{r=1}^{k} \sum_{s=1}^{k}  c_r c_s  \tau_{sk}'\dmat_{12} \tau_{rk}  
    \\ = & \sum_{r=1}^{k}  c_r  \tau_{rk}'\dmat_{12}  \tau_{rk} \left(\sum_{s=1}^{k} c_s\right) - \left(\sum_{s=1}^{k}  c_s  \tau_{sk}\right)'\dmat_{12} \left(\sum_{r=1}^{k}c_r  \tau_{rk} \right)
    \\ =& 0 - {\tau^*} '\dmat_{12} \tau^* 
    \\ \geq& 0
\end{align*}
where the second to last line uses $\sum_{s=1}^{k} c_s=0$ and the definition ${\tau^*}:=\sum_{s=1}^{k}  c_s  \tau_{sk}$. The last line follows because $\dmat_{12}$ is negative semidefinite.
\end{proof}

\subsection{A novel proof of the Aronow-Samii bound}\label{section.ASbound}

Consider an identified bound proposed by \cite{aronowsamii17} that has the a unusual virtue of being perfectly general, i.e., applicable to arbitrary (identified) designs.

\begin{definition}[Aronow-Samii variance bound]
	The ``Aronow-Samii variance bound" is the bound corresponding to the bounding matrix,
	\begin{align*}
	\tilde{\dmat}^\tas:=\dmat +\bfI\left(\dmat=-1\right)+ \diagSmallBracket{
		\bfI\left(\dmat=-1\right) \ones{kn}
	}
	\end{align*}
	where the indicator function, $\textnormal{I}(\dmat=-1)$, returns a matrix of with ones indicating the location of -1 entries in $\dmat$ and zeros elsewhere, and $\diagSmallBracket{.}$ creates a diagonal matrix from a vector.
\end{definition}
\begin{theorem}
	The Aronow-Samii variance bound, $ n^{-2} y' \tilde{\dmat}^\tas y$,  is an identified bound for $ n^{-2}  y' {\dmat} y$.
\end{theorem}
\begin{proof}
	By definition of $\tilde{\dmat}^\tas$, 
	\begin{align*}
	\tilde{\dmat}^\tas-\dmat= \bfI\left(\dmat=-1\right)+\text{\scriptsize diag}\left(\bfI\left(\dmat=-1\right) 1_{\scriptscriptstyle kn}\right).
	\end{align*}
	Note that by construction ($\tilde{\dmat}^\tas-\dmat$) has diagonal elements set equal to the sum of the off-diagonal elements in its row (which by construction are either 0 or 1). The Gershgorin circle theorem implies that a real matrix is positive semi-definite if, for all $i$, the $i^{th}$ diagonal element is greater or equal to the sum of the absolute values of the other elements in the $i^{th}$ row. So, by the Gershgorin circle theorem $\tilde{\dmat}^\tas-\dmat$ is positive semidefinite. Therefore, by Lemma (\ref{psd}), $\nNegTwo y'\tilde{\dmat}^\tas y$ is a variance bound. Moreover, as long as the design is an identified design (i.e., $0<\pi_{1i}<1$ for all $i$), it is an identified bound because $\bfI\left(\dmat=-1\right)$ ensures that the elements of $\dmat$ equal to $-1$ correspond to 0's in $\tilde{\dmat}^{\tas}$.
\end{proof}

\begin{remark}
    Aronow and Samii (2017) derive their bound using Young's inequality. The above-theorem and proof using the Gershgorin circle theorem tie their insight to the current framework.  
\end{remark}



\subsection{Proposed algorithm for variance bounds for any design}\label{section.Mbound}

The following is an algorithm which that can obtain an identified variance bound. Like the AS bound it has the virtue of being applicable to virtually any design. The algorithm is a proof of concept, demonstrating the utility of the notation scheme which allows for the application of matrix theory for the creation of alternative bounds. The subject of comparing bounds will be considered further in Section \ref{section.comparing.bounds}.

\begin{algorithm}\label{algorithm}
\end{algorithm}
\begin{enumerate}
	\item Initialize $kn \times kn$ matrix $\mathbf{t}$. Examples could be $\bfI(\dmat=-1)$ or, if the conditions for the Neyman bound not be applicable, start with $\tilde{\dmat}^\tgn-\dmat$ which may approximate a bound
	\item Obtain the eigen decomposition of matrix $\mathbf{t}$. If all eigenvalues are non-negative (within tolerance), goto Step 6, otherwise continue
	\item Update $\mathbf{t}=\mathbf{v} (\mathbf{e} \circ \bfI(\mathbf{e}>0)) \mathbf{v}'$ where $\mathbf{v}$ is the matrix of eigenvectors and $\mathbf{e}$ is a diagonal matrix of eigenvalues
	\item Update $\mathbf{t}=\bfI(\dmat=-1)+\bfI(\dmat\neq -1) \circ \mathbf{t}$
	\item Return to Step 2
	\item Set $\tilde{ \dmat }^\tm =\dmat+\mathbf{t}$
\end{enumerate}	
As above, $\circ$ is elementwise multiplication and, for example, $\bfI(\mathbf{e}>0)$ is an indicator function returning a matrix of ones and zeros indicating which elements of $\mathbf{e}$ are greater than zero.

Conceptually, the goal of the algorithm is to create a matrix $\mathbf{t}$ that can be added to $\dmat$ yielding a $\tilde{ \dmat }$ matrix that corresponds to an identified variance bound. By Lemma \ref{psd} and Definition \ref{def.identified.bound}, there are two requirements for $\mathbf{t}$. First it must be positive semi-definite, and, second, elements corresponding to $-1$'s in the matrix $\dmat$ must equal one. In step 1, $\mathbf{t}$ meets the second criterion, but not the first.  In step 3, the algorithm creates an approximation to the initial $\mathbf{t}$ matrix by way of the eigen decomposition that ensures positive semi-definiteness, thus meeting the first criterion. However, due to the approximation, $\mathbf{t}$ no longer meets the second criterion. Therefore, in step 4 the algorithm forces $\mathbf{t}$ to have 1's wherever $\dmat$ has $-1$'s in order to again meet the second criteria. But doing so means that $\mathbf{t}$ will no longer meet the first criteria. So, the algorithm iterates through steps 2-4 until convergence is achieved (i.e., until all eigenvalues are non-negative in step 2) at which point $\mathbf{t}$ meets both criteria and, thus, $\tilde{ \dmat }^\tm$ corresponds to an identified bound.

\subsection{Comparing bounds}\label{section.comparing.bounds}

\begin{definition}[Tighter bound]
	Let $\tilde{\dmat}^a$ and $\tilde{\dmat}^b$ correspond to two identified bounds. Matrix $\tilde{\dmat}^a$ is corresponds to a \textit{tighter bound} than $\tilde{\dmat}^b$ if $\tilde{ \dmat }^b - \tilde{ \dmat }^a$ is positive semidefinite. 
\end{definition}

\begin{definition}[Invariant bounding matrix]\label{def.invariant.bound}
	A matrix $\dtilde$ is an \textit{invariant bounding matrix} if it is an bounding matrix and if all $n \times n$ partitions, $\dtilde_{ij} \ones{kn} =  0_{\scriptscriptstyle kn}$, i.e., all rows of the partition (or, equivalently, all columns) sum to zero.
\end{definition}

\exampleType{Paired randomization}  Consider a pair-randomized design, whereby units are ``blocked" (i.e., stratified) into groups of two, and then, in each block, one of the two units is randomly assigned to treatment while the other is assigned to control. Assignments across blocks are independent. 

When $n=4$ (and assuming w.l.o.g. that the data are sorted by pair), the matrix $\dmat$ is
\begin{align*}
\dmat = \left[\begin{smallmatrix}
 \hspace{2.2mm}1 & -1 & & & -1 & \hspace{2.2mm}1 &  
\\ -1 & \hspace{2.2mm}1 & & & \hspace{2.2mm}1 & -1
\\ & & \hspace{2.2mm}1 & -1 & & & -1 & \hspace{2.2mm}1
\\ & & -1 & \hspace{2.2mm}1 & & & \hspace{2.2mm}1 & -1
\\ -1 & \hspace{2.2mm}1 & & & \hspace{2.2mm}1 & -1
\\ \hspace{2.2mm}1 & -1 & & & -1 & \hspace{2.2mm}1 & 
\\ & & -1 & \hspace{2.2mm}1 & & & \hspace{2.2mm}1 & -1
\\ & & \hspace{2.2mm}1 & -1 & & & -1 & \hspace{2.2mm}1 & 
\end{smallmatrix}\right],
\end{align*}
noting that empty cells represent 0. 

For the pair-randomized design, the Neyman bound cannot be applied because $\dmat_{00}$ and $\dmat_{11}$ have negative entries. The Aronow-Samii bound and Algorithm \ref{algorithm} have bounding matrices
\begin{center}
$\dtilde^\tas = \left[\begin{smallmatrix}
\hspace{2.2mm}3 &  & & &  & \hspace{2.2mm}1 &  
\\  & \hspace{2.2mm}3 & & & \hspace{2.2mm}1 & 
\\ & & \hspace{2.2mm}3 & & & & & \hspace{2.2mm}1
\\ & &  & \hspace{2.2mm}3 & & & \hspace{2.2mm}1 & 
\\  & \hspace{2.2mm}1 & & & \hspace{2.2mm}3 & 
\\ \hspace{2.2mm}1 &  & & &  & \hspace{2.2mm}3 & 
\\ & &  & \hspace{2.2mm}1 & & & \hspace{2.2mm}3 & 
\\ & & \hspace{2.2mm}1 & & & &  & \hspace{2.2mm}3 & 
\end{smallmatrix}\right]$, \hspace{2mm} and \hspace{2mm} 
$\dtilde^\tm = {\left[\begin{smallmatrix}
\hspace{2.2mm}2 &  & & &  & \hspace{2.2mm}2 &  
\\  & \hspace{2.2mm}2 & & & \hspace{2.2mm}2 & 
\\ & & \hspace{2.2mm}2 & & & & & \hspace{2.2mm}2
\\ & &  & \hspace{2.2mm}2 & & & \hspace{2.2mm}2 & 
\\  & \hspace{2.2mm}2 & & & \hspace{2.2mm}2 & 
\\ \hspace{2.2mm}2 &  & & &  & \hspace{2.2mm}2 & 
\\ & &  & \hspace{2.2mm}2 & & & \hspace{2.2mm}2 & 
\\ & & \hspace{2.2mm}2 & & & &  & \hspace{2.2mm}2 & 
\end{smallmatrix}\right]}$,
\end{center}
\noindent respectively.  By the Gershgorian circle theorem the difference,
\begin{align*}
\dtilde^\tas - \dtilde^\tm= \left[\begin{smallmatrix}
\hspace{2.2mm}1 &  & & &  & -1 &  
\\  & \hspace{2.2mm}1 & & & -1 & 
\\ & & \hspace{2.2mm}1 & & & & & -1
\\ & &  & \hspace{2.2mm}1 & & & -1 & 
\\  & -1 & & & \hspace{2.2mm}1 & 
\\ -1 &  & & &  & \hspace{2.2mm}1 & 
\\ & &  & -1 & & & \hspace{2.2mm}1 & 
\\ & & -1 & & & &  & \hspace{2.2mm}1 & 
\end{smallmatrix}\right],
\end{align*}
is positive semi-definite, proving that $\dtilde^\tm$ corresponds to a tighter variance bound. Confirmation also comes from eigendecomposition of the difference, $\dtilde^\tas - \dtilde^\tm$, which yields all non-negative eigenvalues: 2, 2, 2, 2, 0, 0, 0, and 0. 

One might alternatively choose the invariant bounding matrix,
\begin{align*}
\dtilde^{\tinv} = {\left[\begin{smallmatrix}
	\hspace{1mm}2 &  & -1 &-1 &  & \hspace{1mm}2 &-1 & -1  
	\\  & \hspace{1mm}2 & -1&-1 & \hspace{1mm}2 & &-1 & -1 
	\\ -1 & -1 & \hspace{1mm}2 & & -1 & -1 & & \hspace{1mm}2
	\\ -1 & -1 &  & \hspace{1mm}2 & -1& -1& \hspace{1mm}2 & 
	\\  & \hspace{1mm}2 &-1 & -1 &  \hspace{1mm}2 & &-1 & -1  
	\\ \hspace{1mm}2 &  & -1 & -1 &  & \hspace{1mm} 2  &-1 & -1 
	\\ -1 & -1 &  & \hspace{1mm}2 & -1 & -1 & \hspace{1.2mm}2 & 
	\\ -1 & -1 & \hspace{1.2mm}2 & & -1  & -1  &  & \hspace{1.2mm}2 & 
	\end{smallmatrix}\right]}.
\end{align*}
The bound can be verified because the eigenvalues of $\dtilde^\tinv-\dmat$ are 8, 0, 0, 0, 0, 0, 0, and 0. However, eigendecomposition of $\dtilde^\tinv-\dtilde^\tm$ gives eigenvalues 4, 0, 0, 0, 0, 0, 0, and -4, indicating that the better bound may depend on outcome vector, $y$, and perhaps the estimator as well.
{ \hfill $\triangle$}

\section{Variance bound estimation}\label{section.EstimatingBound}

With an identified variance bounds defined and several methods of obtaining matrices, $\tilde{\dmat}$, this section turns to the subject of variance bound \textit{estimation}. 

First define the $kn \times kn$ matrix of probabilities and joint probabilities of assignment, 
\begin{align}
{\matp} := \E \left[\R 1_{\scriptscriptstyle kn} 1'_{\scriptscriptstyle kn} \R \right] \nonumber.
\end{align}
Next define an inverse probability weighted version of bounding matrix, $\tilde{ \dmat }$, as
\begin{align} 
\dtildep  := \tilde{ \dmat } / \matp
\end{align}
with $/$ denoting element-wise division defined such that division by zero equals zero.  Then an unbiased estimator of a variance bound for the Horvitz-Thompson estimator can be written,
\begin{align}\label{var_ht_est}
\widehat{\tilde{\V}}\left(\widehat{\delta}^{\tht} \right) :=  {z_c^\tht}' \R \dtildep \R z_c^\tht,
\end{align}
with $z_c^\tht := \diagSmallBracket{y}\onesmat \left(\onesmat' \onesmat \right)^{-1} c$. 
It is unbiased for the variance bound ${z_c^\tht}' \dtilde z_c^\tht$ because $\E \left[\R \dtildep \R \right]=\tilde{ \dmat }$ by construction.  Being inverse-probability weighted, the variance bound estimator in (\ref{var_ht_est}) is, itself, a Horvitz-Thompson estimator.

For other linear estimators, examples of which are given in Table \ref{table.examples.linear.ests}, the bound $z_c' \tilde{ \dmat }z_c$ cannot be estimated unbiasedly because the definition of $z_c$ will often include quantities that, themselves, must be estimated.  However, an appeal to the plug-in principle suggests the use of
\begin{align}\label{var_lin_est}
\widehat{\tilde{\V}}\left(\widehat{\delta}^{\tl }_c \right) := \widehat{z}_c' \R \dtildep \R \widehat{z}_c
\end{align}
with $\widehat{z}_c$ having the same form as $z_c$ but with sample analogues replacing some components. 

\exampleType{The special case of Eicker-Huber-White (a.k.a. ``heteroskedastic consistent", ``sandwich", and ``robust") standard errors} For the OLS estimator, $z^\tols_c$ is defined in Table (\ref{table.examples.linear.ests}). The plug-in principle motivates the use of
\begin{align*}
	\R \widehat{z}^\tols_c = \bpi  \diag{\R  \widehat{u}} \xx \left( \xx' \R \xx\right)^{-1} c,
\end{align*}
where $\R \widehat{u}:=\R (y-\xx \widehat{b}^{\tols} )$ and $\widehat{b}^{\tols}:=\left( \xx' \R \xx\right)^{-1} \xx' \R y$ is the OLS coefficient. Then from equation (\ref{var_lin_est}) we have,
\begin{align*}
\widehat{\tilde{\V}}\left(\widehat{\delta}^{\tl(\tols) }_c \right) 
	= & c' \left(\xx' \R \xx \right)^{-1} \xx' \diag{\R \widehat{u}} \bpi \dtildep \bpi     \diag{\R  \widehat{u}}  \xx \left(\xx' \R \xx \right)^{-1}\hspace{-1mm}c.
\end{align*}
This is a the variance bound estimator in (\ref{var_lin_est}) made specific to OLS. So far it is applicable to virtually any design and any variance bound.

Next, specify a Bernoulli design, in which units are assigned independently to treatment. (Probabilities of assignment may be equal across units, but they need not be in this example.) In this design, the diagonal elements of $\dmat$ are equal to the diagonal of $\bpi^{-1} -\I_{\scriptscriptstyle kn}$, where $\I_{\scriptscriptstyle kn}$ is an identity matrix. Further, any of the above bounding methods yields $\tilde{\dmat}=\bpi^{-1} -\I + \I=\bpi^{-1}$. Thus $\dtildep=\bpi^{-2}$ so that $\bpi \dtildep \bpi=\I_{\scriptscriptstyle kn}$ is the identity matrix. So the OLS variance bound estimator for Bernoulli designs simplifies to,
\begin{align*}
\widehat{\tilde{\V}}{}^{\tb}\left(\widehat{\delta}_c^{\tl(\tols) } \right) = & c' \left(\xx' \R \xx \right)^{-1}\xx' \diag{\R \widehat{u}^2}  \xx \left(\xx' \R \xx \right)^{-1} c.
\end{align*}
This is White's (1980) canonical ``sandwich" variance estimator, sometimes referred to as HC0. { \hfill $\triangle$}

\begin{remark}
    The example shows that Eicker-Huber-White standard errors are a special case of (\ref{var_lin_est}) for OLS in a Bernoulli design.  Note, however, that (\ref{var_lin_est}) is much more general. It applies to any linear estimator, virtually any design and any (identified) variance bound.
\end{remark} 

\begin{remark}
	Adjustments for degrees of freedom (e.g., HC1) or leverage (e.g., HC2, HC3, etc.) can be applied as well.  
\end{remark}

\exampleType{The special case of ``cluster robust" standard errors} Also consider this variance bound estimator for OLS in designs in which clusters are assigned independently to treatment.  Then, if we choose the Neyman bound $\dtildep^\tgn$ (or $\dtildep^\tm$, which is equivalent in the case of for Bernoulli assignment of clusters), and assuming w.l.o.g. that units are sorted by cluster, then $\bpi \dtildep^{\tgn} \bpi $ resolves to a block diagonal matrix of 1's with the blocks corresponding to clusters. Hence, (\ref{var_lin_est}) also reproduces the ``cluster-robust" standard errors sometimes referred to as CR0 as a special case.{ \hfill $\triangle$}

\section{Asymptotics}

\subsection{Conditions for Convergence of Horvitz-Thompson Estimators}

First establishing the unbiasedness of Horvitz-Thompson estimators will allow for straightforward proofs of consistency.

\begin{lemma}\label{lemma.HT.unbiased}
The Horvitz-Thompson estimator for an outcome vector, $y$, and given contrast, $c$, is unbiased for $\delta_c$.
\end{lemma}
\begin{proof}
    \begin{align*}
        \E \left[ c' \left(\onesmat' \onesmat \right)^{-1} \onesmat' \bpiInv \R y \right]= & c' \left(\onesmat' \onesmat \right)^{-1} \onesmat' \bpiInv \E \left[ \R \right]y 
        \\ = & c' \left(\onesmat' \onesmat \right)^{-1} \onesmat' \bpiInv \bpi y
        \\ = &c' \left(\onesmat' \onesmat \right)^{-1} \onesmat'y
        \\ = & \delta_c,
    \end{align*}
\end{proof}

\begin{condition}[Bounded contrast]\label{condition.bound.c}
The chosen contrast vector is finite, i.e., there exists a finite value $u_c$ such that $\textnormal{max}(|c|)< u_c$.
\end{condition}

\begin{condition}[Bounded outcomes]\label{condition.bound.y}
There exists a finite value, $u_y$, such that $\textnormal{max}(|y|)<u_y$ for all $n$.
\end{condition}

\begin{condition}[Design constraint for consistent Horvitz-Thompson estimators]\label{condition.bound.dmat}
There exists a finite value, $u_\dmat$, such that $n^{-1} ||\dmat||_{{1,1}} < u_\dmat$ for all $n$, where $||.||_{1,1}$ is the matrix norm that sums the absolute values of the matrix entries.
\end{condition}

\begin{theorem}[Root-n consistency of HT estimators]\label{theorem.HT.consistency}
By Lemma \ref{lemma.HT.unbiased} and Conditions \ref{condition.bound.c}-\ref{condition.bound.dmat} the Horvitz-Thompson estimator is root-n consistent.    
\end{theorem}
\begin{proof}
    Given Lemma \ref{lemma.HT.unbiased}, it is sufficient to show that the variance converges at the parametric rate, i.e., that $n \V(\delta_c^\tht)$ is bounded, in order to prove consistency.  By Holder's Inequality,
    \begin{align*}
        n \V(\delta_c^\tht) & \leq  n \hspace{1mm}\textnormal{max}(|c'\left(\onesmat' \onesmat \right)^{-1}\onesmat'y|)^2  ||\dmat||_{{1,1}}
        \\ & \leq \textnormal{max}(|c|)^2 \textnormal{max}(|y|)^2 n^{-1} ||\dmat||_{{1,1}}
        \\ & \leq u_c^2 u_y^2 u_\dmat ,
    \end{align*}
    with the last line using Conditions \ref{condition.bound.c}-\ref{condition.bound.dmat}.
\end{proof}

\exampleType{Checking consistency of HT estimators for completely randomized experiments}   Consider a completely randomized experiment with $n$ units where a fixed number of units, $n_c$, are randomly assigned to control, and the remainder, $n_t= n - n_c$, are assigned to treatment. Assume an asymptotic sequence of designs is such that there exists a constant value, $\pi$, such that $\frac{n_t}{n} \rightarrow \pi$ as $n \rightarrow \infty$ with $0<\pi_t<1$. Partition $\dmat$ into four $(n \times n)$ matrices and let $\dmat_{ab}$ represent the $a,b \in \{1,2\}$ partition. Each partition has elements which take on two possible values, one on the diagonal and another on the off-diagonal.  Because $\dmat_{12}=\dmat_{21}$, entries of matrix $\dmat$ take on one of six possible values. In Table \ref{table.d.completerand}, analysis of the these six values and their corresponding frequencies shows that a completely randomized design yields $\frac{1}{n}||\dmat||_{{1,1}}=2\left(\frac{n_t}{n_c}+\frac{n_c}{n_t}+2\right)=O(1)$. Thus, Condition \ref{condition.bound.dmat} is satisfied. Therefore, by Theorem \ref{theorem.HT.consistency}, Horvitz-Thompson estimators are consistent for completely randomized experiments for bounded contrast, $c$, and outcome vector, $y$. { \hfill $\triangle$}

\begin{table}
    \centering
    \begin{tabular}{|c|c|c|c|c|}\hline
     partition  & $ij$ pattern & count &  $\scriptstyle \{\dmat_{ab}\}_{ij}$ & count$\scriptstyle \times \frac{1}{n} \{\dmat_{ab}\}_{ij} $   \\
     \hline $\dmat_{11}$   & $\scriptstyle i=j$ & $\scriptstyle n$ & $\scriptstyle \frac{n_t}{n_c}$   & $\scriptstyle \frac{n_t}{n_c}=  O(1)$
                 \\        & $\scriptstyle i \neq j$ & $\scriptstyle n(n-1)$ & $\scriptstyle -\frac{n_t}{n_c(n-1)}$   & $\scriptstyle -\frac{n_t}{n_c}=  O(1)$
    \\ \hline $\dmat_{12}$ or $\dmat_{21}$   & $\scriptstyle  i=j$ & $\scriptstyle  2n$ & $\scriptstyle  -1$   & $\scriptstyle  -2=O(1)$
                 \\        & $\scriptstyle i \neq j$ & $\scriptstyle 2n(n-1)$ & $\scriptstyle \frac{1}{(n-1)}$   & $\scriptstyle  2=O(1)$
     \\ \hline $\dmat_{22}$   & $\scriptstyle i=j$ & $\scriptstyle n$ & $\scriptstyle \frac{n_c}{n_t}$   & $\scriptstyle \frac{n_c}{n_t}=\scriptstyle O(1)$
                 \\        & $\scriptstyle i \neq j$ & $\scriptstyle n(n-1)$ & $\scriptstyle -\frac{n_c}{n_t(n-1)}$   & $\scriptstyle -\frac{n_c}{n_t}= O(1)$ 
                 \\ \hline
    \end{tabular}
    \caption{Analysis of Condition \ref{condition.bound.dmat} for Complete Randomization}
    \label{table.d.completerand}
\end{table} 

\subsection{Conditions for Convergence of WLS estimator class}

\begin{condition}[Bounded covariates]\label{condition.bound.x}
There exists a finite value $u_\x$ that bounds the covariate values, i.e., $\textnormal{max}(|\x |)<u_\x $, for all $n$.
\end{condition}

\begin{condition}[Bounded $\bpi \m$]\label{condition.bound.m}
There exists a finite value $u_{\bpi\m}$ that bounds $\bpi$ times the WLS ``weighting" matrix $\m$, i.e., $\textnormal{max}(| \bpi \m |)<u_{\bpi\m} $, for all $n$.
\end{condition}

\begin{lemma}[Root-n consistency of WLS]\label{lemma.WLS.consistency}
By conditions \ref{condition.bound.y}-\ref{condition.bound.m} and Theorem \ref{theorem.HT.consistency}, the WLS ``numerator'' vector, $\frac{1}{n} \xx' \m \R y$, is root-n consistent for $\frac{1}{n} \xx' \m \bpi y$. Likewise, by conditions \ref{condition.bound.dmat}-\ref{condition.bound.m} and Theorem \ref{theorem.HT.consistency}, the WLS ``denominator'' matrix, $\frac{1}{n} \xx' \m \R \xx$, is root-n consistent for $\frac{1}{n} \xx' \m \bpi \xx$.  Further, by the continuous mapping theorem $\widehat{b}^{\twls} \rightarrow b^\twls$.
\end{lemma}

\begin{proof}
       Let $\xx_i$ be the column vector created from the $i^{th}$ column of $\xx$. Then the $i^{the}$ element of $\frac{1}{n} \xx \m \R y $ can be written,    
    \begin{align*}
        \{ \frac{1}{n} \xx \m \R y \}_{i } =& \frac{1}{n} \xx_i' \m \R y
        \\ = & \frac{1}{n} \ones{kn}' \diag{\xx_i} \m \R y
        \\ = & \frac{1}{n} \ones{kn}' \R \m \diag{\xx_i}  y
        \\ = & \ones{k}' \w^{\tht} \R  \bpi \m \diag{\xx_i} y
        \\ = & \ones{k}' \w^{\tht} \R  q
    \end{align*}
    with $q=\bpi \m \diag{\xx_i} y $, showing that the elements of the denominator matrix are Horvitz-Thompson estimators with outcome vector $q$ and contrast vector $\ones{k}$. Now by Theorem \ref{lemma.HT.unbiased}, this is consistent because $q$ is bounded, i.e., $\textnormal{max}(|q|) \leq u_{\bpi \m} u_\x u_y $.
    
Similarly, the $i,j$ element of the WLS ``denominator'' matrix, can be written
    \begin{align*}
        \{ \frac{1}{n} \xx \m \R \xx \}_{ij} = & \ones{k}' \w^{\tht} \R  r
    \end{align*}
    with $r=\bpi \m \diag{\xx_i}  \xx_j  $, showing that the elements of the denominator matrix are Horvitz-Thompson estimators with outcome vector $r$ and contrast vector $\ones{k}$.  Now by Theorem \ref{lemma.HT.unbiased}, this is consistent because $r$ is bounded, i.e., $\textnormal{max}(|r|) \leq u_{\bpi \m} u_\x^2 $.  
\end{proof}

\begin{condition}[Stability of WLS ``denominator" estimand]\label{condition.denom.stable}
The denominator of the ``true" WLS coefficient, $\frac{1}{n} \xx' \m \pi \xx$, is invertable for all $n$ and converges in probability to a matrix, $\mathbf{v}$, with finite entries.
\end{condition}

\begin{theorem}[Consistency of the Taylor approximation]
By Conditions \ref{condition.bound.y}-\ref{condition.denom.stable}, the Taylor approximate coefficient, $ \widehat{b}^{\tl(\twls)}:= b^{\twls} + \left(\xx' \m \bpi \xx \right)^{-1}\xx' \m \R (y-\xx b^{\twls})$ is root-n consistent for ${b}^{\twls}:= \left(\xx' \m \bpi \xx \right)^{-1}\xx' \m \bpi y$, i.e., $ \widehat{b}^{\tl(\twls)}- b^{\twls} = O_p(1/\sqrt{n})$.
\end{theorem}
\begin{proof}
    We have
\begin{align*}
     \widehat{b}^{\tl(\twls)}- b^{\twls}:= &\left(\xx' \m \bpi \xx \right)^{-1}\xx' \m \R (y-\xx b^{\twls})
    \\ =& \left(\xx' \m \bpi \xx \right)^{-1} \left( \left(\xx' \m\R y-\xx' \m \bpi y\right) -\left( \xx' \m \R \xx -\xx' \m \bpi \xx \right) b^{\twls} \right)
        \\ =& \left(\frac{1}{n} \xx' \m \bpi \xx \right)^{-1} \left( \left(\frac{1}{n}  \xx' \m\R y-\frac{1}{n}\xx' \m \bpi y\right) -
        \left( \frac{1}{n}\xx' \m \R \xx - \frac{1}{n} \xx' \m \bpi \xx \right) b^{\twls} \right) 
        \\ =& O_p(1) \left( O_p(1/ \sqrt{n})+O_p(1/ \sqrt{n}) \right),
        \\=& O_p(1/ \sqrt{n})
\end{align*}
where the second to last line uses Lemma \ref{lemma.WLS.consistency} and Condition \ref{condition.denom.stable}.
\end{proof}

\begin{theorem}[Asymptotic Equivalence of WLS and its Taylor Approximation]
By Lemma \ref{lemma.WLS.consistency} and Condition \ref{condition.denom.stable}, WLS is asymptotically equivalent to the Taylor linear approximation for WLS.
\end{theorem}
\begin{proof}
    Let the ``true" WLS coefficient be $b^\twls = \left(\xx' \m \bpi \xx \right)^{-1}\xx' \m \bpi y$, then write the WLS estimator as,
    \begin{align*}
        c' \W^\twls \R y = & c' \left( \xx' \m \R \xx \right)^{-1} \xx' \m \R y
        \\ = & c'  b^{\twls} + c' \left( \xx' \m \R \xx \right)^{-1} \xx' \m \R \left( y - \xx b^{\twls}\right)
        \\ = & c'  b^{\twls} + c' \left( \xx' \m \bpi \xx \right)^{-1} \xx' \m \R \left( y - \xx b^{\twls}\right) 
        \\ & \hspace{9mm}+ c' \left( \left( \xx' \m \R \xx \right)^{-1} -\left( \xx' \m \bpi \xx \right)^{-1} \right)\xx' \m \R \left( y - \xx b^{\twls}\right)
        \\ = & c'  b^{\twls} + c' \left( \xx' \m \bpi \xx \right)^{-1} \xx' \m \R \left( y - \xx b^{\twls}\right) 
        \\ & \hspace{9mm}+ c' \left( \left( \xx' \m \R \xx \right)^{-1} -\left( \xx' \m \bpi \xx \right)^{-1} \right)\xx' \m \R \left( \xx \widehat{b}^{\twls} - \xx b^{\twls}\right)
        \\ & \hspace{9mm}+ c' \left( \left( \xx' \m \R \xx \right)^{-1} -\left( \xx' \m \bpi \xx \right)^{-1} \right)\xx' \m \R \left( y- \xx \widehat{b}^{\twls} \right)
                \\ = & \widehat{\delta}_c^{\tl(\twls)} + c' \left( \left( \frac{1}{n}\xx' \m \R \xx \right)^{-1} -\left( \frac{1}{n} \xx' \m \bpi \xx \right)^{-1} \right)\left(\frac{1}{n}\xx' \m \R \xx \right) \left( \widehat{b}^{\twls} - b^{\twls}\right)
        \\ = & \widehat{\delta}_c^{\tl(\twls)} + O_p(1/\sqrt{n})O_p(1)O_p(1/\sqrt{n})
        \\ = & \widehat{\delta}_c^{\tl(\twls)} + O_p(1/n)
     \end{align*}
\end{proof}

\subsection{Conditions for consistent variance estimation}
 
\begin{definition}[Second-order design matrix]
  The ``second-order design matrix" is a forth order tensor $(kn \times kn \times kn \times kn)$ of variances and covariances of inverse-probability weighted pairwise-joint inclusion indicators, written,
  \begin{align*}
      \dubd :=\Big(  \E\left[\left(\R 1_{\scriptscriptstyle n}1_{\scriptscriptstyle kn}'\R\right) \otimes \left( \R 1_{\scriptscriptstyle kn}1_{\scriptscriptstyle kn}'\R\right) \right]  -\matp \otimes \matp \Big) / \left( \matp \otimes \matp \right),
  \end{align*}
  where $\matp:=\E\left[\R 1_{\scriptscriptstyle n}1_{\scriptscriptstyle kn}'\R\right]$ is a matrix of with inclusion probabilities on the diagonal and pair-wise joint inclusion probabilities off the diagonal, ``$\otimes$" is the tensor outer product and ``/" is elementwise division with division by zero resolving to zero.
\end{definition}

\begin{condition}[Second order design constraint for consistent variance estimation]\label{condition.bound.dd}
There exists a finite constant $u_{\dubd}$ such that $\frac{1}{n} \left|\left|\left(\dtilde \otimes \dtilde \right) \circ \dubd \right|\right|_{1,1,1,1} < u_{\dubd}$ for all $n$,
where ``$\otimes$" is tensor outer product, ``$\circ$" is elementwise multiplication, $\left|\left|. \right|\right|_{1,1,1,1}$ gives the sum of the absolute values of the tensor entries.
\end{condition}

\begin{theorem}[Consistency of the Horvitz-Thompson variance estimator]
By Conditions \ref{condition.bound.c}-\ref{condition.bound.dmat} and \ref{condition.bound.dd} the variance estimator for the Horvitz-Thompson point estimator is consistent.
\end{theorem}

\begin{proof}
The variance of $n$ times the Horvitz-Thompson variance estimator (times $n$) is,
\begin{align*}
n \V \bigg ( n \widehat{\tilde{\V}}\left(\widehat{\delta}^{\tht} \right) \bigg ) = & n\E \left[ \left(n  {z^\tht}' \R \dtildep \R z^\tht - n {z^\tht}' \tilde{\dmat} z^\tht \right)^2\right]
\\ \leq  & \max\left(\left|c\right|\right)^4 \max\left(\left|y\right|\right)^4 \frac{1}{n} \left| \left|\left(\dtilde \otimes \dtilde \right) \circ \dubd \right|\right|_{1,1,1,1}
\\ \leq  & u_c^4 u_y^4 u_{\dubd},
\end{align*}
where the second line uses Holder's inequality and the last line uses Conditions \ref{condition.bound.c}, \ref{condition.bound.y} and \ref{condition.bound.dd}.
\end{proof}

\exampleType{Checking consistency of HT variance (bound) estimator for completely randomized experiments} Again consider a completely randomized experiment with $n$ units where a fixed number of units, $n_c$, are randomly assigned to control, and the remainder, $n_t= n - n_c$, are assigned to treatment. Assume an asymptotic sequence of designs is such that there exists a constant value, $\pi$, such that $\frac{n_t}{n} \rightarrow \pi$ as $n \rightarrow \infty$ with $0<\pi<1$. Let the variance bound be the Neyman bound, i.e., $\dtilde =\dtilde^\tgn$ from Definition \ref{definition.neyman.bound}. Table \ref{table.analysis.dd.condition} enumerates the unique values that appear in $\frac{1}{n} \left|\left|\left(\dtilde^\tgn \otimes \dtilde^\tgn \right) \circ \dubd \right|\right|_{1,1,1,1}$ along with their relative frequencies and shows that Condition \ref{condition.bound.dd} is satisfied for completely randomized experiments. Hence, the Horvitz-Thompson variance (bound) estimator given in Equation \ref{var_ht_est} is consistent for the Neyman bound. { \hfill $\triangle$}


\begin{landscape}
\centering

\begin{table}
\centering 
\begin{tabular}{|c|c|c|c|c|c| c|} \hline 
 \makecell{Treatment (T) \\ or Control (C)} & $ijkl$ pattern & count &  $\makecell{ \hspace{-8mm}\{ \scriptstyle \E\left[\left(\R 1_{\scriptscriptstyle n}1_{\scriptscriptstyle kn}'\R\right) \right.  \\    \left. \scriptstyle \otimes   \left(\R 1_{\scriptscriptstyle n}1_{\scriptscriptstyle kn}'\R\right)\right] \}_{\scriptscriptstyle ijkl}}$ &  $\{ \matp \otimes \matp\}_{\scriptscriptstyle ijkl}$  & $\{\dtilde^\tgn \otimes \dtilde^\tgn \}_{\scriptscriptstyle ijkl}$ 
 & count $\times  \scriptstyle \frac{1}{n} \big \{ \left(\dtilde \otimes \dtilde \right) \circ \dubd \big \} _{\scriptscriptstyle ijkl}
$
\\ \hline $\scriptstyle i,j,k,l \in \textnormal{C}$  &  $\scriptstyle i=j=k=l$ & $\scriptstyle n$
            &  $\frac{n_c}{n}$ & $\frac{n_c^2}{n^2}$
             & $\frac{n^2}{n_c^2}$ & $\scriptstyle \frac{n^2 n_t }{n_c^3}=O(1)$
     \\ &   \makecell{$\scriptstyle i=j=k,l$ or $\scriptstyle i=j=l,k$ \vspace{-1.5mm} \\  or $\scriptstyle i,j=k=l$ or $\scriptstyle i=k=l,j$}  &  $\scriptstyle 4n(n-1)$ & $\frac{n_c (n_c-1)}{n(n-1)}$ & $\frac{n_c^2(n_c-1)}{n^2(n-1)}$
         & $-\frac{n^2}{n_c^2(n-1)}$ & $\scriptstyle -\frac{4n^2 n_t}{n_c^3}=O(1)$ 
     \\ &  $\scriptstyle i=j, k=l$  &  $\scriptstyle n(n-1)$ & $ \frac{n_c (n_c-1)}{n(n-1)}$ & $\frac{n_c^2 }{n^2 }$
        & $\frac{n^2}{n_c^2}$ & $\scriptstyle \frac{n_t n^2}{n_c^3}=O(1)$
     \\ &  $\scriptstyle i=j,k,l$ or $\scriptstyle i,j,k=l$  & $\scriptstyle 2n(n-1)(n-2)$ & $ \frac{n_c (n_c-1)(n_c-2)}{n(n-1)(n-2)}$ & $\frac{n_c^2(n_c-1) }{n^2(n-1) }$
        & $-\frac{n^2}{n_c^2(n-1)}$ & $\scriptstyle  \frac{4n_t n^2}{ n_c^3}=O(1)$
\\ &  $\scriptstyle i=k,j=l$ or $\scriptstyle i=l, j=k$ & $\scriptstyle 2n(n-1)$  &
          $ \frac{n_c (n_c-1)}{n(n-1)}$ & $\frac{n_c^2(n_c-1)^2 }{n^2(n-1)^2 }$
          & $-\frac{n^2}{n_c^2(n-1)^2}$ & $\scriptstyle -\frac{2 n^2}{n_c^2}=O(1)$
     \\ &   \makecell{$\scriptstyle i=k, j, l$ or $\scriptstyle i=l,j,k$ \vspace{-1.5mm} \\ or $\scriptstyle i,j=k,l$ or $\scriptstyle i,j=l,k$} & $\scriptstyle 4n (n-1)(n-2)$ 
        & $ \frac{n_c (n_c-1)(n_c-2)}{n(n-1)(n-2)}$ & $\frac{n_c^2(n_c-1)^2  }{n^2(n-1)^2  }$
        & $\frac{n^2}{n_c^2(n-1)^2}$ & $\scriptstyle \frac{4 n^2}{n_c^2}\left(\frac{n(n_c-2)}{n_c(n_c-1)}-\frac{(n-2)}{(n-1)} \right)=O(1)$
     \\ & $\scriptstyle i,j,k,l$ & $\scriptscriptstyle n(n-1)(n-2)(n-3)$ & 
        $ \scriptstyle \frac{n_c (n_c-1)(n_c-2)(n_c-3)}{n(n-1)(n-2)(n-3)}$ & $\scriptstyle \frac{n_c^2(n_c-1)^2  }{n^2(n-1)^2  }$
        & $\scriptstyle \frac{n^2}{n_c^2(n-1)^2}$ & \makecell{$\scriptstyle \hspace{-6mm} \frac{-4  n_t n^3}{n_c^2(n-1)(n_c-1)} -\frac{6n^2}{n_c^2(n-1)} $ \\ \hspace{15mm} $\scriptstyle   +\frac{6n^3}{n_c^3(n_c-1)}=O(1) $}
\\ \hline  $\scriptstyle i,j \in \textnormal{T}$ $\scriptstyle k,l\in \textnormal{C}$ &    $\scriptstyle i=j=k=l$ & $\scriptstyle 2n$ & $\scriptstyle 0$  & $\scriptstyle \frac{n_c n_t}{n^2}$ & $\scriptstyle \frac{n^2 }{n_t n_c}$ & $\scriptstyle -\frac{n^2}{n_t n_c}=O(1)$
     \\ or $\scriptstyle i,j \in \textnormal{C}$ $\scriptstyle k,l\in \textnormal{T}$ &  \makecell{$\scriptstyle i=j=k,l$ or $\scriptstyle i=j=l,k$ \vspace{-1.5mm} \\ or $\scriptstyle i,j=k=l$ or $\scriptstyle i=k=l,j$} &  $\scriptstyle 8n(n-1)$ & $\scriptstyle 0$    & $\scriptstyle \frac{n_c n_t (n_t-1)}{n^2(n-1)}$ & $\scriptstyle-\frac{n^2}{n_t n_c(n-1)}$ & $\scriptstyle -\frac{8n^2}{n_t n_c}=O(1)$
     \\ &  $\scriptstyle i=j,k=l$  & $\scriptstyle 2 n (n-1)$ &  $\frac{n_t n_c}{n^2}$  & $\frac{n_c n_t}{n^2}$ & $\frac{n^2 }{n_t n_c}$ & $\scriptstyle 0$
     \\ &   $\scriptstyle i=j,k,l$  &  $\scriptstyle 4 n(n-1)(n-2)$ & $\frac{n_t n_c(n_t-1)}{n^2(n-1)}$  & $\frac{n_c n_t (n_t-1)}{n^2(n-1)}$ & $ \frac{n^2}{n_t n_c(n-1)}$  & $\scriptstyle 0$
\\ &   $\scriptstyle i=k,j=l$ or $\scriptstyle i=l,j=k$  & $\scriptstyle 4n(n-1)$  & $\scriptstyle 0$ & $\frac{n_c n_t (n_t-1) (n_c-1)}{n^2(n-1)^2}$ & $\frac{n^2}{n_t n_c(n-1)^2}$ & $\scriptstyle  \frac{4 n^2}{n_t n_c(n-1)}=O(1/n)$ 
     \\ &  $\scriptstyle i=k=l,j$ or $\scriptstyle i,j=k=l$    & $\scriptstyle 4n(n-1)$ & $\scriptstyle 0$ & $\frac{n_c n_t (n_c-1)}{n^2(n-1)}$ & $-\frac{n^2}{n_t n_c(n-1)}$ & $\scriptstyle  -\frac{4 n^2}{n_t n_c}=O(1)$ 
     \\ &  \makecell{$\scriptstyle i=k, j, l$ or $\scriptstyle i=l,j,k$ \vspace{-1.5mm} \\ or $\scriptstyle i,j=k,l$ or $\scriptstyle i,j=l,k$} & $\scriptstyle 8n (n-1)(n-2)$  & $\scriptstyle 0$  & $\frac{n_c n_t (n_c-1) (n_t-1)}{n^2(n-1)^2}$ & $\frac{n^2}{n_t n_c(n-1)^2}$ & $\scriptstyle  \frac{8 n^2 (n-2)}{n_t n_c (n-1) }=O(1)$ 
     \\ &  $\scriptstyle i,j,k,l$  & $\scriptscriptstyle 2 n(n-1)(n-2)(n-3)$    
        & $\frac{n_t (n_t-1) n_c (n_c-1)}{n(n-1)(n-2)(n-3)}$  & $\frac{n_c n_t (n_c-1)(n_t-1)}{n^2(n-1)^2}$ & $\frac{n^2}{n_t n_c(n-1)^2}$ & $\scriptstyle  \frac{4 n^2 (n-3)}{n_t n_c (n-1) }=O(1)$ 
\\ \hline $\scriptstyle i,j,k,l \in \textnormal{T}$ &  $\scriptstyle i=j=k=l$ & $\scriptstyle n$
            &  $\frac{n_t }{n}$ & $\frac{n_t ^2}{n^2}$
             & $\frac{n^2}{n_t ^2}$ & $\scriptstyle \frac{n^2 n_c}{n_t ^3}=O(1)$
     \\ &   \makecell{$\scriptstyle i=j=k,l$ or $\scriptstyle i=j=l,k$ \vspace{-1.5mm} \\ or $\scriptstyle i,j=k=l$ or $\scriptstyle i=k=l,j$}  &  $\scriptstyle 4n(n-1)$ & $\frac{n_t  (n_t -1)}{n(n-1)}$ & $\frac{n_t ^2(n_t -1)}{n^2(n-1)}$
         & $-\frac{n^2}{n_t ^2(n-1)}$ & $\scriptstyle -\frac{4n^2 n_c}{n_t ^3}=O(1)$ 
     \\ &  $\scriptstyle i=j, k=l$  &  $\scriptstyle n(n-1)$ & $ \frac{n_t  (n_t -1)}{n(n-1)}$ & $\frac{n_t ^2 }{n^2 }$
        & $\frac{n^2}{n_t ^2}$ & $\scriptstyle \frac{n_c n^2}{n_t ^3}=O(1)$
     \\ &  $\scriptstyle i=j,k,l$ or $\scriptstyle i,j,k=l$  & $\scriptstyle 2n(n-1)(n-2)$ & $ \frac{n_t  (n_t -1)(n_t -2)}{n(n-1)(n-2)}$ & $\frac{n_t ^2(n_t -1) }{n^2(n-1) }$
        & $-\frac{n^2}{n_t ^2(n-1)}$ & $\scriptstyle  \frac{4n_c n^2}{ n_t ^3}=O(1)$
\\ &  $\scriptstyle i=k,j=l$ or $\scriptstyle i=l, j=k$ & $\scriptstyle 2n(n-1)$  &
          $ \frac{n_t  (n_t -1)}{n(n-1)}$ & $\frac{n_t ^2(n_t -1)^2 }{n^2(n-1)^2 }$
          & $-\frac{n^2}{n_t ^2(n-1)^2}$ & $\scriptstyle -\frac{2 n^2}{n_t ^2}=O(1)$
     \\ &   \makecell{$\scriptstyle i=k, j, l$ or $\scriptstyle i=l,j,k$ \vspace{-1.5mm} \\ or $\scriptstyle i,j=k,l$ or $\scriptstyle i,j=l,k$} & $\scriptstyle 4n (n-1)(n-2)$ 
        & $ \frac{n_t  (n_t -1)(n_t -2)}{n(n-1)(n-2)}$ & $\frac{n_t ^2(n_t -1)^2  }{n^2(n-1)^2  }$
        & $\frac{n^2}{n_t ^2(n-1)^2}$ & $\scriptstyle \frac{4 n^2}{n_t ^2}\left(\frac{n(n_t -2)}{n_t (n_t -1)}-\frac{(n-2)}{(n-1)} \right)=O(1)$
     \\ & $\scriptstyle i,j,k,l$ & $\scriptscriptstyle n(n-1)(n-2)(n-3)$ & 
        $ \scriptstyle \frac{n_t  (n_t -1)(n_t -2)(n_t -3)}{n(n-1)(n-2)(n-3)}$ & $\scriptstyle \frac{n_t ^2(n_t -1)^2  }{n^2(n-1)^2  }$
        & $\scriptstyle \frac{n^2}{n_t ^2(n-1)^2}$ & \makecell{$\scriptstyle \hspace{-6mm} \frac{-4  n_c n^3}{n_t^2(n-1)(n_t-1)} -\frac{6n^2}{n_t^2(n-1)} $ \\ \hspace{15mm} $\scriptstyle   +\frac{6n^3}{n_t^3(n_t-1)}=O(1) $}
        \\ \hline 
 \end{tabular}
 \caption{Analysis of Condition \ref{condition.bound.dd} for Complete Randomization}
\end{table}\label{table.analysis.dd.condition}

\end{landscape}

\pagebreak
\begin{appendix}
	\section{Notation Index}
\hspace{-6mm}
	\begin{tabular}{| C{2cm} | L{12cm} | } \hline &
		\\ $n$ & Number of units in the finite population in the experiment
		\\ &	\\  	$1_{\scriptscriptstyle kn}$  & Length-$kn$ column vector of 1's. In matrix notation, serves as a replacement for the more common summation symbol, $\Sigma$
		\\ & \\$y_{0i}$, $y_{1i}$ & The control and treatment potential outcomes for the $i^{th}$ unit, respectively
		\\ & \\  $y_{0}$, $y_{1}$ & Length-$n$ vectors of control and treatment potential outcomes, respetively
		\\ & \\ $y$ & Length-$kn$ vector of all potential outcomes. The first $n$ elements are control potential outcomes multiplied by $-1$, followed by the treatment potential outcomes. Multiplication of control potential outcomes by $-1$ allows for the compact representation of the ATE as the sum of the elements of this vector divided by $n$
		\\ &	\\  	$\delta$ & Average treatment effect (ATE), the parameter of interest
		\\ & \\ $R_{0i}$, $R_{1i}$ & Random indicators of the $i^{th}$ unit's assignment to control and treatment, respectively
		\\ & \\ $R_{0}$, $R_{1}$ & Length-$n$ vectors of assignment indicators for control and treatment, respectively
		\\ & \\ $\R$ &$kn \times kn$ diagonal matrix of assignment indicators. The first $n$ diagonal elements represent the control indicators, followed by $n$ treatment indicators
		\\ & \\ $\pi_{0i}$, $\pi_{1i}$ & For the $i^{th}$ unit, the probability of assignment to control and treatment, respectively
		\\ & \\ $\pi_{0}$, $\pi_{1}$ & Length-$n$ vectors of probabilities of assignment to control and treatment, respectively
		\\ &  \\	$\bpi$ & $kn\times kn$ diagonal matrix of assignment probabilities. The first $n$ diagonal elements give the control probabilities, followed by the treatment probabilities
		\\ & \\  $\pi_{0i0j}$, $\pi_{0i1j}$,   & Joint assignment probabilities for units $i$ and $j$. For example, $\pi_{1i0j}$ is the probability that 
		\\ $\pi_{1i0j}$, $\pi_{1i1j}$    & $i$ is in treatment and $j$ is in control
		\\ & \\ $\dmat$ & $kn\times kn$ ``design" matrix that gives the variance-covariance matrix of the vector $1'_{\scriptscriptstyle kn}\bpi^{-1}\R$. Allows for compact representation of variance of HT estimators as a quadratic in matrix form
		\\ & \\ $\dmat_{00}$, $\dmat_{01}$, & The four $n\times n$ partitions of the matrix $\dmat$. For example, the top-right partition, $\dmat_{01}$, has
		\\ $\dmat_{10}$, $\dmat_{11}$ &  $i,j$ element $\frac{\pi_{0i1j}-\pi_{0i}\pi_{1j}}{\pi_{0i}\pi_{1j}}$
		\\ & \\ $\tilde{\dmat}$ & A modified version of $\dmat$ that allows for compact representation of a variance {\it bound} for HT estimators as a quadratic in matrix form. While the variance of the HT estimator is not identified, a variance bound may be
		\\ & \\ 
		\hline
	\end{tabular}

\pagebreak
\noindent \hspace{-6mm}
\begin{tabular}{| C{2cm}| L{12cm} | } \hline &		
		\\ & \\ $\matp$ & $kn\times kn$ ``probability" matrix that gives the joint assignment probabilities
		\\ & \\ $\matp_{00}$, $\matp_{01}$, & The four $n\times n$ quadrants of the matrix $\matp$. For example, $\matp_{01}$ has $ij$ element $\pi_{0i1j}$
		\\ $\matp_{10}$, $\matp_{11}$ & 		
		\\ & \\ $\tilde{\matp}$ & A modified version of $\matp$ that replaces zeros with ones. Allows for division by $\tilde{\matp}$ without division-by-zero error 
		\\ & \\ $x_i$ & Length-$k$ vector of covariates associated with the $i^{th}$ unit	
		\\ &  \\ $\x$ & An $n \times k$ matrix of covariates	
		\\ &  \\ $\tx$ & An $n \times (k+1)$ matrix representing the concatenation of an intercept vector, $1_n$, and $\x$
		\\ &  \\ $\xx$ & A $kn \times l$ matrix of covariates. The first $n$ rows are multiplied by $-1$ to mirror the vector $y$. Represents an arbitrary specification
		\\ &  \\ $\xx_{\sI}$ & A $kn \times (k+2) $ matrix of covariates. The ``common slopes" specification. Elements in the first $n$ rows are multiplied by -1 to mirror the vector $y$
		\\ &  \\ $\xx_{\II}$ & A $kn \times (2k+2)$ matrix of covariates. The ``separate slopes" specification. Elements in the first $n$ rows are multiplied by -1 to mirror the vector $y$
		\\ & \\ 
		\hline
	\end{tabular}
	\pagebreak
\section{Supplementary Proofs}
\printproofs
\end{appendix}

\end{document}